\documentclass[a4paper]{llncs}

\usepackage{amsmath}
\usepackage{amssymb}
\usepackage{ntheorem}
\usepackage{pifont}
\usepackage{rotating}
\usepackage{tikz}
	\usetikzlibrary{patterns}
	\usetikzlibrary{arrows}
\usepackage{todonotes}
\usepackage{wasysym}
\usepackage{wrapfig}
\input xy
\xyoption{all}

\theoremstyle{plain}\theoremseparator{.}
	\renewtheorem{lemma}{Lemma}[section]
	\renewtheorem{proposition}[lemma]{Proposition}
	\renewtheorem{theorem}[lemma]{Theorem}
	\renewtheorem{corollary}[lemma]{Corollary}
\theorembodyfont{\upshape}
	\renewtheorem{definition}[lemma]{Definition}
\theoremheaderfont{\it}\theorembodyfont{\upshape}
	\renewtheorem{remark}[lemma]{Remark}
	\renewtheorem{example}[lemma]{Example}
\newenvironment{Proof}{\begin{proof}}{\hfill\qed\end{proof}}
\renewcommand{\thelemma}{\arabic{section}.\arabic{lemma}}

\newcommand{\Cat}[1]{\text{\tt{#1}}}			
\newcommand{\cat}[1]{\text{\tt{#1}}}	
\newcommand{\Sets}{\Cat{Set}}				
\newcommand{\SMod}[1]{\text{$\mathbb #1$-\Cat{SMOD}}}	
\newcommand{\Mod}[1]{\text{$\mathbb #1$-\Cat{MOD}}}	
\newcommand{\Vect}[1]{\text{$\mathbb #1$-\Cat{VEC}}}	
\newcommand{\Pca}{\Cat{PCA}}				
\newcommand{\Cone}{\Cat{CONE}}				
\newcommand{\CMon}{\Cat{CMON}}				
\newcommand{\Ab}{\Cat{AB}}				
\newcommand{\CoCat}[1]{{\sf Coalg}({#1})}		

	\newcommand{\after}{\mathrel{\circ}}
	
\newcommand{\Fun}[1]{\ensuremath{#1}}			
\newcommand{\Gh}{\Fun{\widehat F}}			
\newcommand{\Cub}[1]{\Fun{F_{\,{#1}}}}			

\newcommand{\Out}[1]{{#1}_{\text{\rm o}}}		
\newcommand{\Comp}[2]{{#1}_{#2}}			
\DeclareMathOperator{\tr}{tr}				
\newcommand{\alg}[1]{\mathbb{#1}}          
\newcommand{\Min}[1]{\ensuremath{\mu_{#1}}}		
\DeclareMathOperator{\spn}{span}			
\DeclareMathOperator{\co}{co}				
\newcommand{\Ts}{T_{\mathbb{S}}}
\newcommand{\Te}{T_{\mathbb{E}}}
\newcommand{\EM}[1]{\Cat{Set}^{#1}}
\newcommand{\Dis}{\mathcal{D}}				

\DeclareMathOperator{\id}{id}				
\newcommand{\Deli}{{.\kern3pt}}				
\newenvironment{Nili}{
	\begin{list}{}{\leftmargin=0pt\itemsep=5pt\listparindent=\parindent}}{\end{list}}
\DeclareMathOperator{\supp}{supp}

\begin{document}

\title{Proper Semirings and Proper Convex Functors}

\author{%
	Ana Sokolova\inst{1} 
\and 
	Harald Woracek\inst{2}
}
\authorrunning{Ana Sokolova \and Harald Woracek}

\institute{%
	University of Salzburg, Austria\\
	\email{ana.sokolova@cs.uni-salzburg.at}
\and
	TU Vienna, Austria\\
	\email{harald.woracek@tuwien.ac.at}
}

\maketitle

\begin{abstract}
	Esik and Maletti introduced the notion of a proper semiring and proved that
	some important (classes of) semirings -- Noetherian semirings, natural numbers
	-- are proper. Properness matters as the equivalence problem for weighted
	automata over a semiring which is proper and finitely and effectively presented is
	decidable. 
	Milius generalised the notion of properness from a semiring to a functor. As a
	consequence, a semiring is proper if and only if its associated ``cubic
	functor'' is proper. Moreover, properness of a functor renders soundness and
	completeness proofs for axiomatizations of equivalent behaviour. 

	In this paper we provide a method for proving properness of functors, and
	instantiate it to cover both the known cases and several novel ones: (1)
	properness of the semirings of positive rationals and positive reals, via
	properness of the corresponding cubic functors; and (2) properness of two
	functors on (positive) convex algebras. The latter functors are important for
	axiomatizing trace equivalence of probabilistic transition systems. Our proofs
	rely on results that stretch all the way back to Hilbert and Minkowski.  
	\keywords{proper semirings, proper functors, coalgebra, weighted automata, probabilistic transition systems}
\end{abstract}


%
\section{Introduction}
\label{sec-Intro}
In this paper we deal with algebraic categories and deterministic weighted automata functors on them. 
Such categories are the target of generalized determinization~\cite{silva.bonchi.bonsangue.rutten:2010,silva.sokolova:2011,jacobs.silva.sokolova:2015} and enable coalgebraic modelling beyond sets. For example, non-deterministic automata, weighted, or probabilistic ones are coalgebraically modelled over the categories of join-semilattices, semimodules for a semiring, and convex sets, respectively. Moreover, expressions for axiomatizing behavior semantics often live in algebraic categories. 

In order to prove completeness of such axiomatizations, the common approach~\cite{silva:2010,bonsangue.milius.silva:2011,silva.sokolova:2011} is to prove finality of a certain object in a category of coalgebras over an algebraic category. Proofs are significantly simplified if it suffices to verify finality only w.r.t. coalgebras carried by free finitely generated algebras, as those are the coalgebras that result from generalized determinization.  

In recent work, Milius~\cite{milius:2017} proposed the notion of a proper functor on an algebraic category that provides a sufficient condition for this purpose. This notion is an extension of the notion of a proper semiring introduced by Esik and Maletti~\cite{esik.maletti:2010}: A semiring is proper if and only if its ``cubic" functor is proper. 
A cubic functor is a functor $\mathbb S \times (-)^A$ where $A$ is a finite alphabet and $\mathbb S$ is a free algebra with a single generator in the algebraic category. Cubic functors model deterministic weighted automata which are models of determinizations of non-deterministic and probabilistic transition systems.


Properness is the property that for any two states that are behaviourally equivalent in coalgebras with free finitely generated carriers, there is a zig-zag of homomorphisms (called a chain of simulations in the original works on weighted automata and proper semirings) that identifies the two states and whose nodes are all carried by free finitely generated algebras. 

Even though the notion of properness is relatively new for a semiring and very new for a functor, results on properness of semirings can be found in more distant literature as well. Here is a brief history, to the best of our knowledge:
\begin{itemize}
\item The Boolean semiring was proven to be proper in~\cite{bloom.esik:1993}.
\item Finite commutative ordered semirings were proven to be proper in	~\cite[Theorem 5.1]{esik.kuich:2001}. Interestingly, the proof provides a zig-zag with at most seven intermediate nodes.
\item Any euclidean domain and any skew field were proven proper in \cite[Theorem 3]{beal.lombardy.sakarovitch:2005}. In each case the zig-zag has two intermediate nodes.
\item The semiring of natural numbers $\mathbb N$, the Boolean semiring $\mathbb B$, the ring of integers $\mathbb Z$ and any skew field were proven proper in~\cite[Theorem 1]{beal.lombardy.sakarovitch:2006}. Here, all zig-zag were spans, i.e., had a single intermediate node with outgoing arrows.
\item Noetherian semirings were proven proper in~\cite[Theorem 4.2]{esik.maletti:2010}, commutative rings also in~\cite[Corollary 4.4]{esik.maletti:2010}, and finite semirings as well in~\cite[Corollary 4.5]{esik.maletti:2010}, all with a zig-zag being a span. Moreover, the tropical semiring is not proper, as proven in~\cite[Theorem 5.4]{esik.maletti:2010}.
\end{itemize}

Having properness of a semiring, together with the property of the semiring being finitely and effectively presentable, yields decidability of the equivalence problem (decidability of trace equivalence) for weighted automata. 

In this paper, motivated by the wish to prove properness of a certain functor $\Gh$ on convex algebras used for axiomatizing trace semantics of probabilistic systems in~\cite{silva.sokolova:2011}, as well as by the open questions stated in~\cite[Example~3.19]{milius:2017}, we provide a framework for proving properness. We instantiate this framework on known cases like Noetherian semirings and $\mathbb N$ (with a zig-zag that is a span), and further prove new results of properness: 
\begin{itemize}
\item	The semirings $\mathbb Q_+$  and $\mathbb R_+$ of non-negative rationals and reals, respectively, are proper. 
	The shape of the zig-zag is a span as well. 
\item The  functor $[0,1] \times (-)^A$ on \Pca\ is proper, again the zig-zag being a span.
\item The functor $\Gh$ on $\Pca$ is proper. This proof is the most involved, and interestingly, provides the only case where the zig-zag is not a span: it contains three intermediate nodes of which the middle one forms a span. 
\end{itemize}

Our framework requires a proof of so-called \emph{extension} and
\emph{reduction lemmas} in each case. While the extension lemma is a generic
result that covers all cubic functors of interest, the reduction lemma is in
all cases a nontrivial property intrinsic to the algebras under consideration.
For the semiring of natural numbers it is a consequence of a result that we
trace back to Hilbert; for the case of convex algebra $[0,1]$ the result is due
to Minkowski. In the case of $\Gh$, we use Kakutani's set-valued fixpoint theorem.

It is an interesting question for future work whether these new properness results may lead to new complete axiomatizations of expressions for certain weighted automata. 

The organization of the rest of the paper is as follows. In
Section~\ref{sec-ProperFunctors} we give some basic definitions and introduce
the semirings, the categories, and the functors of interest.
Section~\ref{sec-PropernessCubic} provides the general framework as well as
proofs of properness of the cubic functors.
Section~\ref{sec-SubcubicFunctor}--Section~\ref{sec-PropernessSubcubic} lead us
to properness of $\Gh$ on \Pca. For space reasons, we present the ideas of
proofs and constructions in the main paper and defer all detailed proofs to the
appendix. 

\paragraph{Acknowledgements.} We thank the anonymous reviewers for many valuable comments, in particular for reminding us of a categorical property that shortened the proof of the extension lemma (the proofs of Lemma~\ref{lem:ext-lem-1} and Lemma~\ref{lem:ext-lem-2} in Appendix~\ref{app:B}).

\section{Proper functors}
\label{sec-ProperFunctors}
We start with a brief introduction of the basic notions from algebra and
coalgebra needed in the rest of the paper, as well as the important definition
of proper functors~\cite{milius:2017}. We refer the interested reader to~\cite{Rut00:tcs,JR96:eatcs,jacobs:2017}
for more details. We assume basic knowledge of category theory, see e.g.~\cite{maclane:1998} or Appendix~\ref{sec:app-basics}.

Let $\Cat C$ be a category and $F$ a $\Cat C$-endofunctor.
The category $\CoCat F$ of \emph{$F$-coalgebras} is the category 
having as objects pairs $(X,c)$ where $X$ is an object of $\Cat C$ and 
$c$ is a $\Cat C$-morphism from $X$ to $FX$, and as morphisms $f\colon(X,c)\to(Y,d)$ those $\Cat C$-morphisms from 
$X$ to $Y$ that make the diagram on the right commute. 
 
\begin{wrapfigure}{r}{.1\textwidth}\vspace*{-.80cm}
		{\xymatrix@R=0.7em@C=0.8em{
		X \ar[rr]^f \ar[d]_c && Y \ar[d]^d
		\\
		FX \ar[rr]^{Ff} && FY
	}}\vspace*{-.80cm}
\end{wrapfigure}
All base categories $\Cat C$ in this paper will be \emph{algebraic categories}, i.e., categories $\EM{T}$ of 
Eilenberg-Moore algebras of a finitary monad~\footnote{The notions of monads
and algebraic categories are central to this paper. We recall them in
Appendix~\ref{sec:app-basics} to make the paper accessible to all readers.} in
$\Sets$. Hence, all base categories are concrete with forgetful functor that is
identity on morphisms.  

In such categories behavioural equivalence~\cite{Kurz00:thesis,Wol00:cmcs,Staton11} can be defined as follows.
Let $(X,c)$ and $(Y,d)$ be $F$-coalgebras and let $x\in X$ and $y\in Y$. 
Then $x$ and $y$ are \emph{behaviourally equivalent}, and we write $x\sim y$, if there exists an $F$-coalgebra $(Z,e)$ and 
$\CoCat F$-morphisms $f\colon (X,c)\to(Z,e)$, $g\colon(Y,d)\to(Z,e)$, with $f(x)=g(y)$. 
\[
	\xymatrix@C=40pt{
		(X,c) \ar[r]^f
		& (Z,e) 
		\save[]+<0pt,-13pt>*\txt{$\scriptstyle f(x)=g(y)$}\restore
		& (Y,d) \ar[l]_g 
	}
\]
If there exists a final coalgebra in $\CoCat F$, and all functors considered in this paper will have this property, 
then two elements are behaviourally equivalent if and only if they have the
same image in the final coalgebra. If we have a \emph{zig-zag diagram} in $\CoCat F$
\begin{equation}\label{zig-zag}
	\xymatrix@R=0.7em@C=9pt{
		(X,c) \ar[rd]^{f_1} && (Z_2,e_2) \ar[ld]_{f_2} \ar[rd]^{f_3} && 
		\cdots \ar[ld]_{f_4} \ar[rd]^{f_{2n-1}} && (Y,d) \ar[ld]_{f_{2n}} &
		\\
		& (Z_1,e_1) 
		&& (Z_3,e_1) 
		&& (Z_{2n-1},e_1) 
		&
	}
\end{equation}
which relates $x$ with $y$ in the sense that there exist elements $z_{2k}\in Z_{2k}$, $k=1,\ldots,n-1$, 
with (setting $z_0=x$ and $z_{2n}=y$) 
\[
	f_{2k}(z_{2k})=f_{2k-1}(z_{2k-2}),\quad k=1,\ldots, n
	,
\]
then $x\sim y$. 

We now recall the notion of a proper functor, introduced by Milius~\cite{milius:2017} which is central to this paper. It is very helpful for establishing completeness of regular expressions calculi, 
cf.\ \cite[Corollary~3.17]{milius:2017}. 

\begin{definition}\label{Proper}
	Let $T\colon\Sets\to\Sets$ be a finitary monad with unit $\eta$ and multiplication $\mu$. 
	A $\EM{T}$-endofunctor $F$ is \emph{proper}, if the following statement holds.

	For each pair $(TB_1,c_1)$ and $(TB_2,c_2)$ of $F$-coalgebras with $B_1$ and $B_2$ finite sets, 
	and each two elements $b_1\in B_1$ and $b_2\in B_2$ with $\eta_{B_1}(b_1)\sim\eta_{B_2}(b_2)$, 
	there exists a zig-zag \eqref{zig-zag} in $\CoCat F$ which relates $\eta_{B_1}(b_1)$ with $\eta_{B_2}(b_2)$, 
	and whose nodes $(Z_j,e_j)$ all have free and finitely generated carrier. 
\end{definition}

This notion generalizes the notion of a proper semiring introduced by Esik and Maletti in 
\cite[Definition~3.2]{esik.maletti:2010}, cf.\ \cite[Remark~3.10]{milius:2017}.

\begin{remark}\label{fg}
	In the definition of properness the condition that intermediate nodes have free \emph{and} finitely generated 
	carrier is necessary for nodes with incoming arrows (the nodes $Z_{2k-1}$ in \eqref{zig-zag}). 
	For the intermediate nodes with outgoing arrows ($Z_{2k}$ in \eqref{zig-zag}), it is enough to require 
	that their carrier is finitely generated.
	This follows since every $F$-coalgebra with finitely generated carrier is the image under an $F$-coalgebra morphism 
	of an $F$-coalgebra with free and finitely generated carrier. 

	Moreover, note that zig-zags which start (or end) with incoming arrows instead of outgoing ones, can also be allowed since a zig-zag of this form can be turned into one of the form \eqref{zig-zag} 
	by appending identity maps. 
\end{remark}

\subsection*{Some concrete monads and functors}

We deal with the following base categories.
\begin{itemize}
\item The category $\SMod S$ of semimodules over a semiring $\mathbb S$ induced by the monad $T_{\mathbb S}$ of finitely supported 
	maps into $\mathbb S$, see, e.g., \cite[Example~4.2.5]{manes.mulry:2007}.
\item The category \Pca\ of positively convex algebras induced by the monad of finitely supported 
	subprobability distributions, see, e.g., \cite{doberkat:2006,doberkat:2008} and \cite{pumpluen:1984}.
\end{itemize}
For $n\in\mathbb N$, the free algebra with $n$ generators in $\SMod S$ is the direct product $\mathbb S^n$, and in \Pca\ it is 
the $n$-simplex $\Delta^n=\{(\xi_1,\ldots,\xi_n)\mid \xi_j\geq 0,\sum_{j=1}^n\xi_j\leq 1\}$. 

Concerning semimodule-categories, we mainly deal with the semirings $\mathbb N$, $\mathbb Q_+$, and $\mathbb R_+$, 
and their ring completions $\mathbb Z$, $\mathbb Q$, and $\mathbb R$.
For these semirings the categories of $\mathbb S$-semimodules are 
\begin{itemize}
\item $\CMon$ of commutative monoids for $\mathbb N$,
\item $\Ab$ of abelian groups for $\mathbb Z$,
\item $\Cone$ of convex cones for $\mathbb R_+$,
\item $\Vect Q$ and $\Vect R$ of vector spaces over the field of rational and real numbers, respectively, for 
	$\mathbb Q$ and $\mathbb R$.
\end{itemize}
We consider the following functors, where $A$ is a fixed finite alphabet.
Recall that we use the term \emph{cubic functor} for the functor $T1\times(-)^A$ where $T$ is a monad on $\Sets$. 
We chose the name since $T1\times(-)^A$ assigns to objects $X$ a full direct product, i.e., a full cube. 
\begin{itemize}
\item The \emph{cubic functor} $\Cub{\mathbb S}$ on $\SMod S$, i.e., the functor acting as 
	\begin{align*}
		& \Cub{\mathbb S}X=\mathbb S\times X^A\text{ for }X\text{ object of }\SMod S,
		\\
		& \Cub{\mathbb S}f=\id_{\mathbb S}\times(f\circ -)
		\text{ for }f\colon X\to Y\text{ morphism of }\SMod S.
	\end{align*}
	The underlying $\Sets$ functors of cubic functors are also sometimes called deterministic-automata functors, 
	see e.g.~\cite{jacobs.silva.sokolova:2015}, as their coalgebras are 
	deterministic weighted automata with output in the semiring.
\item The \emph{cubic functor} $\Cub{[0,1]}$ on $\Pca$, i.e., the functor 
	$\Cub{[0,1]}X=[0,1]\times X^A$ and $\Cub{[0,1]}f=\id_{[0,1]}\times(f\circ -)$.
\item A \emph{subcubic convex functor} $\Gh$ on \Pca\ whose action will be introduced in 
	Definition~\ref{Ghat-def}.\footnote{This functor was denoted $\hat G$ in~\cite{silva.sokolova:2011} 
	where it was first studied in the context of axiomatization of trace semantics.}\ 
	The name originates from the fact that $\Gh X$ is a certain convex subset of $\Cub{[0,1]}X$ and that 
	$\Gh f=(\Cub{[0,1]}f)|_{\Gh X}$ for $f\colon X\to Y$.
\end{itemize}
Cubic functors are liftings of $\Sets$-endofunctors, in particular, they preserve surjective algebra homomorphisms. 
It is easy to see that also the functor $\Gh$ preserves surjectivity, cf.\ Lemma~\ref{GhSurj} (Appendix~\ref{app:D})
This property is needed to apply the work of Milius, cf.\ \cite[Assumptions~3.1]{milius:2017}.

\begin{remark}\label{ProperSemiring}
	We can now formulate precisely the connection between proper semirings and proper functors mentioned after Definition~\ref{Proper}. A semiring $\mathbb S$ is proper in the sense of \cite{esik.maletti:2010}, 
	if and only if for every finite input alphabet $A$ the cubic functor $\Cub{\mathbb S}$ on $\SMod S$ is proper. 
\end{remark}

We shall interchangeably think of direct products as sets of functions or as sets of tuples. 
Taking the viewpoint of tuples, the definition of $\Cub{\mathbb S}f$ reads as 
\[
	(\Cub{\mathbb S}f)\big((o,(x_a)_{a\in A})\big)=\big(o,(f(x_a))_{a\in A}\big),\quad 
	o\in\mathbb S,\ x_a\in X\text{ for }a\in A
	.
\]
A coalgebra structure $c\colon X\to\Cub{\mathbb S}X$ writes as 
\[
	c(x)=\big(\Out c(x),(\Comp ca(x))_{a\in A}\big),\quad x\in X
	,
\]
and we use $\Out c:X\to\mathbb S$ and $\Comp ca:X\to X$ as generic notation for the components of the map $c$. 
More generally, we define $c_w\colon X\to X$ for any word $w\in A^*$ inductively as $\Comp c{\varepsilon}=\id_X$ and $\Comp c{wa}=\Comp ca\circ\Comp cw,\ w\in A^*,a\in A$.
	
The map from a coalgebra $(X,c)$ into the final $\Cub{\mathbb S}$-coalgebra, the \emph{trace map}, is then given as 
	$\tr_c(x)=\big((\Out c\circ\Comp cw)(x)\big)_{w\in A^*}$ for $ x\in X$.
Behavioural equivalence for cubic functors is the kernel of the trace map.

\section{Properness of cubic functors}
\label{sec-PropernessCubic}

Our proofs of properness in this section and in Section~\ref{sec-PropernessSubcubic} below start from the following idea. 
Let $\mathbb S$ be a semiring, and assume we are given two $\Cub{\mathbb S}$-coalgebras which have free finitely generated carrier, 
say $(\mathbb S^{n_1},c_1)$ and $(\mathbb S^{n_2},c_2)$. Moreover, assume $x_1\in\mathbb S^{n_1}$ and $x_2\in\mathbb S^{n_2}$ are two 
elements having the same trace. For $j=1,2$, let 
$d_j\colon \mathbb S^{n_1}\times\mathbb S^{n_2}  \to  \Cub{\mathbb S}(\mathbb S^{n_1}\times\mathbb S^{n_2})$ be given by $$d_j(y_1, y_2) = \Big(\Out{c_j}(y_j),((\Comp{c_1}a(y_1),\Comp{c_2}a(y_2)))_{a\in A}\Big).$$

Denoting by $\pi_j\colon \mathbb S^{n_1}\times\mathbb S^{n_2}\to\mathbb S^{n_j}$ the canonical projections, 
both sides of the following diagram separately commute.
\[
\xymatrix@C=15pt@R=10pt@M=7pt{
	\mathbb S^{n_1} \ar[dd]_{c_1}
	&& \mathbb S^{n_1}\times\mathbb S^{n_2} \ar[ll]_{\pi_1} \ar[rr]^{\pi_2} 
		\ar@/_15pt/[dd]_{d_1} \ar@/^15pt/[dd]^{d_2} 
	&& \mathbb S^{n_2} \ar[dd]_{c_2}
	\\
	&& \neq &&
	\\
	\Cub{\mathbb S}\mathbb S^{n_1} 
	&& \Cub{\mathbb S}(\mathbb S^{n_1}\times\mathbb S^{n_2}) \ar[ll]_{\Cub{\mathbb S}\pi_1} \ar[rr]^{\Cub{\mathbb S}\pi_2}
	&& \Cub{\mathbb S}\mathbb S^{n_2}
}
\]
However, in general the maps $d_1$ and $d_2$ do not coincide. 

The next lemma contains a simple observation: 
there exists a subsemimodule $Z$ of $\mathbb S^{n_1}\times\mathbb S^{n_2}$, such that the restrictions of 
$d_1$ and $d_2$ to $Z$ coincide and turn $Z$ into an $\Cub{\mathbb S}$-coalgebra.

\begin{lemma}\label{Prod}
	Let $Z$ be the subsemimodule of $\mathbb S^{n_1}\times\mathbb S^{n_2}$ generated by the pairs $(\Comp{c_1}w(x_1),\Comp{c_2}w(x_2))$ for $w \in A^*$. 
	Then $d_1|_Z=d_2|_Z$ and $d_j(Z)\subseteq\Cub{\mathbb S}(Z)$.
\end{lemma}

The significance of Lemma~\ref{Prod} in the present context is that it leads to the diagram (we denote $d=d_j|_Z$)
\[
\xymatrix@C=15pt@R=35pt@M=7pt{
	\mathbb S^{n_1} \ar[d]_{c_1}
	&& Z \ar[ll]_{\pi_1} \ar[rr]^{\pi_2} \ar[d]^d 
		\save[]+<3pt,10pt>*\txt{\begin{rotate}{90}$\subseteq$\end{rotate}}\restore
		\save[]+<2pt,27pt>*\txt{$\mathbb S^{n_1}\!\!\times\mathbb S^{n_2}$}\restore
	&& \mathbb S^{n_2} \ar[d]_{c_2}
	\\
	\Cub{\mathbb S}\mathbb S^{n_1} 
	&& \Cub{\mathbb S}Z \ar[ll]_{\Cub{\mathbb S}\pi_1} \ar[rr]^{\Cub{\mathbb S}\pi_2}
		\save[]+<-3pt,-10pt>*\txt{\begin{rotate}{-90}$\subseteq$\end{rotate}}\restore
		\save[]+<2pt,-29pt>*\txt{$\mathbb S\!\!\times(\mathbb S^{n_1}\!\!\times\mathbb S^{n_2})^A$}\restore
	&& \Cub{\mathbb S}\mathbb S^{n_2}
}
\]
In other words, it leads to the zig-zag in \CoCat{$\Cub{\mathbb S}$}
\begin{equation}\label{ZZ}
\xymatrix@C=40pt{
	(\mathbb S^{n_1},c_1) & (Z,d) \ar[l]_{\mkern20mu\pi_1} \ar[r]^{\pi_2\mkern20mu} 
	& (\mathbb S^{n_2},c_2) 
}
\end{equation}
This zig-zag relates $x_1$ with $x_2$ since $(x_1,x_2)\in Z$. 
If it can be shown that $Z$ is always finitely generated, it will follow that $\Cub{\mathbb S}$ is proper. 

Let $\mathbb S$ be a Noetherian semiring, i.e., a semiring such that every $\mathbb S$-subsemimodule of some finitely generated 
$\mathbb S$-semimodule is itself finitely generated. Then $Z$ is, as an $\mathbb S$-subsemimodule of 
$\mathbb S^{n_1}\times\mathbb S^{n_2}$ finitely generated. Hence, we reobtain the result 
\cite[Theorem~4.2]{esik.maletti:2010} of Esik and Maletti.

\begin{corollary}[Esik--Maletti 2010]\label{Noetherian}
	Every Noetherian semiring is proper.
\end{corollary}

Our first main result is Theorem~\ref{CubProp} below, where we show properness of the cubic functors 
$\Cub{\mathbb S}$ on \SMod{S}, for $\mathbb S$ being one of the semirings $\mathbb N$, $\mathbb Q_+$, $\mathbb R_+$, 
and of the cubic functor $\Cub{[0,1]}$ on \Pca. 
The case of $\Cub{\mathbb N}$ is known from 
\cite[Theorem~4]{beal.lombardy.sakarovitch:2006}~\footnote{In~\cite{beal.lombardy.sakarovitch:2006} only a 
sketch of the proof is given, cf.\ \cite[\S3.3]{beal.lombardy.sakarovitch:2006}. In this sketch one important point is not mentioned. 
Using the terminology of \cite[\S3.3]{beal.lombardy.sakarovitch:2006}: 
it could a priori be possible that the size of the vectors in $G$ and the size of $G$ both oscillate.}, 
the case of $\Cub{[0,1]}$ is stated as an open problem in \cite[Example~3.19]{milius:2017}.

\begin{theorem}\label{CubProp}
	The cubic functors $\Cub{\mathbb N}$, $\Cub{\mathbb Q_+}$, $\Cub{\mathbb R_+}$, and $\Cub{[0,1]}$ are proper. 

	In fact, for any two coalgebras with free finitely generated carrier and any two elements having 
	the same trace, a zig-zag with free and finitely generated nodes relating those elements can be found, 
	which is a span (has a single intermediate node with outgoing arrows).
\end{theorem}

The proof proceeds via relating to the Noetherian case.
It always follows the same scheme, which we now outline.
Observe that the ring completion of each of $\mathbb N$, $\mathbb Q_+$, $\mathbb R_+$, is Noetherian 
(for the last two it actually is a field), and that $[0,1]$ is the positive part of the unit ball in $\mathbb R$. 

\begin{Nili}
\item \textit{Step 1. The extension lemma:}\ 
	We use an extension of scalars process to pass from the given category \Cat{C} to an associated category 
	\Mod E\ with a Noetherian ring $\mathbb E$. This is a general categorical argument. 

	To unify notation, we agree that $\mathbb S$ may also take the value $[0,1]$, and that $T_{[0,1]}$ is the 
	monad of finitely supported subprobability distributions giving rise to the category \Pca. 
	\begin{center}
	\begin{tabular}{|@{\ }l@{\ }||@{\ }l@{\ }|@{\ }l@{\ }|@{\ }l@{\ }|@{\ }l@{\ }|}
		\hline
		$\mathbb S$ \raisebox{-6pt}{\rule{0pt}{17pt}} & $\mathbb N$ & $\mathbb Q_+$ & $\mathbb R_+$ & $[0,1]$
		\\ \hline
		\Cat{C} \raisebox{-6pt}{\rule{0pt}{17pt}} & \SMod N\ (\CMon) & \SMod{Q_+} & \SMod{R_+} (\Cone) & \Pca
		\\ \hline
		\Mod E \raisebox{-6pt}{\rule{0pt}{17pt}} & \Mod Z (\Ab) & \Mod Q (\Vect Q) & \Mod R (\Vect R) & \Mod R (\Vect R)
		\\ \hline
	\end{tabular}
	\end{center}
For the formulation of the extension lemma, recall that the starting
category \Cat{C} is the Eilenberg-Moore category of the monad $\Ts$ and the
target category \Mod E\ is the  Eilenberg-Moore category of $\Te$. We write
$\eta_{\mathbb S}$ and $\mu_{\mathbb S}$ for the unit and multiplication of
$\Ts$ and analogously for $\Te$. We have $\Ts \le \Te$, via the inclusion monad morphism
$\iota\colon \Ts \Rightarrow \Te$ given by $\iota_X(u) = u$, as 
$\eta_{\mathbb E} = \iota \after \eta_{\mathbb S}$ and $\mu_{\mathbb E} \after \iota\iota =\iota\after \mu_{\mathbb S}$ 
where $\iota\iota \stackrel{\text{def}}{=}\Te\iota\after \iota \stackrel{\text{nat.}}{=} \iota \after \Ts\iota$. Recall that a monad morphism $\iota\colon T_\mathbb S \to T_\mathbb E$ defines a functor $M_\iota\colon \EM{T_\mathbb E} \to \EM{T_\mathbb S}$ which maps a $T_\mathbb E$-algebra $(X, \alpha_X)$ to $(X, \iota_X\after \alpha_X)$ and is identity on morphisms. Obviously, $M_\iota$ commutes with the forgetful functors $U_\mathbb S: \EM{T_\mathbb S} \to \Sets$ and $U_\mathbb E: \EM{T_\mathbb E} \to \Sets$, i.e., $U_\mathbb S \after M_\iota = U_\mathbb E$.

\begin{definition}
	Let $(X, \alpha_X) \in \EM{\Ts}$ and $(Y, \alpha_Y) \in \EM{\Te}$ 
	where $\Ts$ and $\Te$ are monads with $\Ts \le \Te$ via $\iota\colon \Ts \Rightarrow \Te$. 
	A $\Sets$-arrow $h\colon X \to Y$ is a $\Ts \le \Te$-homomorphism from $(X,\alpha_X)$ to $(Y,\alpha_Y)$ 
	if and only if the following diagram commutes (in $\Sets$)
	$$\xymatrix@R=0.7em{
	\Ts X \ar[rr]^{\iota h} \ar[d]_{\alpha_X}&& \Te Y \ar[d]^{\alpha_Y}\\
	X \ar[rr]^{h}&& Y
	}$$
	where $\iota h$ denotes the map 
	$\iota h \stackrel{\text{def}}{=} \Te h\after \iota_X \stackrel{\text{nat.}}{=} \iota_Y \after \Ts h$.	
	In other words, a $\Ts \le \Te$-homomorphism from $(X,\alpha_X)$ to $(Y,\alpha_Y)$ is a morphism in $\EM{T_\mathbb S}$ from $(X,\alpha_X)$ to $M(Y,\alpha_Y)$.  
\end{definition}

Now we can formulate the extension lemma.  

\begin{proposition}[Extension Lemma]\label{prop:ext-lem}
	For every $\Cub{\mathbb{S}}$-coalgebra $\Ts B \stackrel{c}{\to} \Cub{\mathbb{S}}(\Ts B)$ with free finitely generated carrier 
	$\Ts B$ for a finite set $B$, there exists an $\Cub{\mathbb{E}}$-coalgebra 
	$\Te B \stackrel{\tilde c}{\to} \Cub{\mathbb{E}}(\Te B)$ with free finitely generated carrier $\Te B$ such that 
	$$
		\xymatrix@R=0.7em{
			\Ts B \ar[rr]^{\iota_B} \ar[d]_{c}&& \Te B \ar[d]^{\tilde c}\\
			\Cub{\mathbb{S}}(\Ts B) \ar[rr]^{\iota_1 \times (\iota_B)^A}&& \Cub{\mathbb{E}}(\Te B)
		}
	$$
	where the horizontal arrows ($\iota_B$ and $\iota_1 \times \iota_B^A$) are $\Ts \le \Te$-homomorphisms, 
	and moreover they both amount to inclusion.
\end{proposition}

\item \textit{Step 2. The basic diagram:}\ 
	Let $n_1,n_2\in\mathbb N$, let $B_j$ be the $n_j$-element set consisting of the canonical basis vectors of $\mathbb E^{n_j}$, 
	and set $X_j=T_{\mathbb S}B_j$.  
	Assume we are given $\Cub{\mathbb S}$-coalgebras $(X_1,c_1)$ and $(X_2,c_2)$, and elements $x_j\in X_j$ 
	with $\tr_{c_1}x_1=\tr_{c_2}x_2$. 

	The extension lemma provides $\Cub{\mathbb E}$-coalgebras 
	$(\mathbb E^{n_j},\tilde c_j)$ with $\tilde c_j|_{X_j}=c_j$. Clearly, $\tr_{\tilde c_1}x_1=\tr_{\tilde c_2}x_2$.
	Using the zig-zag diagram \eqref{ZZ} in \CoCat{$F_{\mathbb E}$} and appending inclusion maps, 
	we obtain what we call the \emph{basic diagram}.
	In this diagram all solid arrows are arrows in \Mod E, and all dotted arrows are arrows in \Cat{C}. 
	The horizontal dotted arrows denote the inclusion maps, and $\pi_j$ are the restrictions to $Z$ 
	of the canonical projections. 
	\[
	\xymatrix@C=15pt@R=35pt@M=7pt{
		X_1 \ar@{.>}[r] \ar@{.>}[d]_{c_1} & \mathbb E^{n_1} \ar[d]_{\tilde c_1}
		&& Z \ar[ll]_{\pi_1} \ar[rr]^{\pi_2} \ar[d]^d 
			\save[]+<3pt,10pt>*\txt{\begin{rotate}{90}$\subseteq$\end{rotate}}\restore
			\save[]+<2pt,27pt>*\txt{$\mathbb E^{n_1}\!\!\times\mathbb E^{n_2}$}\restore
		&& \mathbb E^{n_2} \ar[d]_{\tilde c_2} & X_2 \ar@{.>}[l] \ar@{.>}[d]^{c_2}
		\\
		\Cub{\mathbb S}X_1 \ar@{.>}[r] & \Cub{\mathbb E}\mathbb E^{n_1} 
		&& \Cub{\mathbb E}Z \ar[ll]_{\Cub{\mathbb E}\pi_1} \ar[rr]^{\Cub{\mathbb E}\pi_2}
			\save[]+<-3pt,-10pt>*\txt{\begin{rotate}{-90}$\subseteq$\end{rotate}}\restore
			\save[]+<2pt,-29pt>*\txt{$\mathbb E\!\!\times(\mathbb E^{n_1}\!\!\times\mathbb E^{n_2})^A$}\restore
		&& \Cub{\mathbb E}\mathbb E^{n_2} & \Cub{\mathbb S}X_2 \ar@{.>}[l]
	}
	\]
	Commutativity of this diagram yields 
	$d\big(\pi_j^{-1}(X_j)\big)\subseteq(\Cub{\mathbb E}\pi_j)^{-1}\big(\Cub{\mathbb S}X_j)$ for $j=1,2$.
	Now we observe the following properties of cubic functors.

	\begin{lemma}\label{CubProperties}
		We have $\Cub{\mathbb E}X\cap\Cub{\mathbb S}Y=\Cub{\mathbb S}(X\cap Y)$. 
		Moreover, if $Y_j\subseteq X_j$, then 
		$(\Cub{\mathbb E}\pi_1)^{-1}(\Cub{\mathbb S}Y_1)\cap(\Cub{\mathbb E}\pi_2)^{-1}(\Cub{\mathbb S}Y_2)
		=\Cub{\mathbb S}(Y_1\times Y_2)$.
	\end{lemma}

	Using this, yields
	\begin{align*}
		d\big(Z\cap(X_1\times X_2)\big)\subseteq &\, 
		\Cub{\mathbb E}Z\cap(\Cub{\mathbb E}\pi_1)^{-1}\big(\Cub{\mathbb S}X_1)
		\cap(\Cub{\mathbb E}\pi_2)^{-1}\big(\Cub{\mathbb S}X_2)
		\\
		= &\, \Cub{\mathbb E}Z\cap\Cub{\mathbb S}(X_1\times X_2)=\Cub{\mathbb S}\big(Z\cap(X_1\times X_2)\big)
		.
	\end{align*}
	This shows that $Z\cap(X_1\times X_2)$ becomes an $\Cub{\mathbb S}$-coalgebra with the restriction 
	$d|_{Z\cap(X_1\times X_2)}$. Again referring to the basic diagram, we have the following zig-zag in 
	\CoCat{$F_{\mathbb S}$} (to shorten notation, denote the restrictions of $d,\pi_1,\pi_2$ to $Z\cap(X_1\times X_2)$ 
	again as $d,\pi_1,\pi_2$):
	\begin{equation}\label{4}
		\xymatrix@C=40pt{
		(X_1,c_1) & \big(Z\cap(X_1\times X_2),d\big) 
			\ar[l]_{\pi_1\mkern50mu} \ar[r]^{\mkern60mu\pi_2} 
		& (X_2,c_2)
		}
	\end{equation}
	This zig-zag relates $x_1$ with $x_2$ since $(x_1,x_2)\in Z\cap(X_1\times X_2)$. 

\item \textit{Step 3. The reduction lemma:}\ 
	In view of the zig-zag \eqref{4}, the proof of Theorem~\ref{CubProp} can be completed by 
	showing that $Z\cap(X_1\times X_2)$ is finitely generated as an algebra in \Cat{C}. 
	Since $Z$ is a submodule of the finitely generated module $\mathbb E^{n_1}\times\mathbb E^{n_2}$ over 
	the Noetherian ring $\mathbb E$, it is finitely generated as an $\mathbb E$-module. The task thus is to 
	show that being finitely generated is preserved when reducing scalars.

	This is done by what we call the \emph{reduction lemma}. Contrasting the extension lemma, the reduction lemma 
	is not a general categorical fact, and requires specific proof in each situation. 

	\begin{proposition}[Reduction Lemma]\label{RedLem}
		Let $n_1,n_2\in\mathbb N$, let $B_j$ be the set consisting of the $n_j$ canonical basis vectors of 
		$\mathbb E^{n_j}$, and set $X_j=T_{\mathbb S}B_j$.  
		Moreover, let $Z$ be an $\mathbb E$-submodule of $\mathbb E^{n_1}\times\mathbb E^{n_2}$. 
		Then $Z\cap(X_1\times X_2)$ is finitely generated as an algebra in \Cat{C}.
	\end{proposition}
\end{Nili}

\section{A subcubic convex functor}
\label{sec-SubcubicFunctor}

Recall the following definition from \cite[p.309]{silva.sokolova:2011}.

\begin{definition}\label{Ghat-def}
	We introduce a functor $\Gh\colon \Pca\to\Pca$. 
	\begin{enumerate}
	\item Let $X$ be a \Pca. Then 
		\begin{align*}
			\Gh X=\Big\{ (o,\phi)\in[0,1] & \times X^A\mid 
			\\
			& \exists\,n_a\in\mathbb N\Deli \exists\,p_{a,j}\in[0,1],x_{a,j}\in X\text{ for }j=1,\ldots,n_a,a\in A\Deli 
			\\
			& o+\sum_{a\in A}\sum_{j=1}^{n_a}p_{a,j}\leq 1,\ \phi(a)=\sum_{j=1}^{n_a}p_{a,j}x_{a,j} \Big\}.
		\end{align*}
	\item Let $X,Y$ be \Pca s, and $f\colon X\to Y$ a convex map. Then $\Gh f\colon \Gh X\to\Gh Y$ is the map 
		$\Gh f=\id_{[0,1]}\times(f\circ -)$.
	\end{enumerate}
\end{definition}

For every $X$ we have $\Gh X\subseteq\Cub{[0,1]}X$, and for every $f\colon X\to Y$ we have 
$\Gh f=(\Cub{[0,1]}f)|_{\Gh X}$. For this reason, we think of $\Gh$ as a \emph{subcubic functor}. 

The definition of \Gh\ can be simplified. 

\begin{lemma}\label{Ghat-simple}
	Let $X$ be a \Pca, then 
	\begin{align*}
		\Gh X=\Big\{ (o,f)\in[0,1]\times X^A\mid\ & 
		\exists\,p_a\in[0,1],x_a\in X\text{ for }a\in A\Deli 
		\\
		& o+\sum_{a\in A}p_a\leq 1,\ f(a)=p_ax_a \Big\}.
	\end{align*}
\end{lemma}

From this representation it is obvious that \Gh\ is monotone in the sense that 
\begin{itemize}
\item If $X_1\subseteq X_2$, then $\Gh X_1\subseteq\Gh X_2$. 
\item If $f_1\colon X_1\to Y_1,f_2\colon X_2\to Y_2$ with $X_1\subseteq X_2,Y_1\subseteq Y_2$ and 
	$f_2|_{X_1}=f_1$, then $\Gh f_2|_{\Gh X_1}=\Gh f_1$.
\end{itemize}
Note that \Gh\ does not preserve direct products. 

For a \Pca\ $X$ whose carrier is a compact subset of a euclidean space, $\Gh X$ can be described 
with help of a geometric notion, namely using the Minkowksi functional of $X$. 
Before we can state this fact, we have to make a brief digression to explain this notion and its properties. 

\begin{definition}\label{Minko}
	Let $X\subseteq\mathbb R^n$ be a \Pca. 
	The \emph{Minkowski functional} of $X$ is the map $\Min X\colon \mathbb R^n\to[0,\infty]$ defined as 
	$\Min X(x)=\inf\{t>0\mid x\in tX\}$, where the infimum of the empty set is understood as $\infty$. 
\end{definition}

Minkowski functionals, sometimes also called \emph{gauge}, are a central and exhaustively studied notion in convex geometry, 
see, e.g., \cite[p.34]{rudin:1991} or \cite[p.28]{rockafellar:1970}.

We list some basic properties whose proof can be found in the mentioned textbooks. 
\begin{enumerate}
\item $\Min X(px)=p\Min X(x)$ for $x\in\mathbb R^n,p\geq 0$,
\item $\Min X(x+y)\leq \Min X(x)+\Min X(y)$ for $x,y\in\mathbb R^n$,
\item $\Min{X\cap Y}(x)=\max\{\Min X(x),\Min Y(x)\}$ for $x\in\mathbb R^n$.
\item If $X$ is bounded, then $\Min X(x)=0$ if and only if $x=0$. 
\end{enumerate}

The set $X$ can almost be recovered from $\Min X$.
\begin{enumerate}
\setcounter{enumi}{4}
\item ${\displaystyle\{x\in\mathbb R^n\mid\Min X(x)<1\}\subseteq X\subseteq\{x\in\mathbb R^n\mid\Min X(x)\leq 1\}}$.
\item If $X$ is closed, equality holds in the second inclusion of 5.
\item Let $X,Y$ be closed. Then $X\subseteq Y$ if and only if $\Min X\geq\Min Y$.
\end{enumerate}

\begin{example}\label{MinkoExa}
	As two simple examples, consider the $n$-simplex $\Delta^n\subseteq\mathbb R^n$ and a convex cone $C\subseteq\mathbb R^n$.
	Then (here $\geq$ denotes the product order on $\mathbb R^n$)
	\[
		\Min{\Delta^n}(x)= 
		\begin{cases}
			\sum_{j=1}^n\xi_j &\hspace*{-3mm},\quad x=(\xi_1,\ldots,\xi_n)\geq 0,
			\\
			\infty &\hspace*{-3mm},\quad \text{otherwise}.
		\end{cases}
		\qquad 
		\Min C(x)=
		\begin{cases}
			0 &\hspace*{-3mm},\quad x\in C,
			\\
			\infty &\hspace*{-3mm},\quad \text{otherwise}.
		\end{cases}
	\]
	Observe that $\Delta^n=\{x\in\mathbb R^n\mid \Min{\Delta^n}(x)\leq 1\}$.
\end{example}

Another illustrative example is given by general pyramids in a euclidean space. 
This example will play an important role later on.

\begin{example}\label{Pyramid}
	For $u\in\mathbb R^n$ consider the set 
	\[
		X=\big\{x\in\mathbb R^n\mid x\geq 0\text{ and }(x,u)\leq 1\big\}
		,
	\]
	where $(\cdot,\cdot)$ denotes the euclidean scalar product on $\mathbb R^n$. 
	The set $X$ is intersection of the cone $\mathbb R_+^n$ with the half-space given by the inequality $(x,u)\leq 1$, 
	hence it is convex and contains $0$. Thus $X$ is a \Pca. 

	Let us first assume that $u$ is strictly positive, i.e., $u\geq 0$ and no component of $u$ equals zero. 
	Then $X$ is a pyramid (in $2$-dimensional space, a triangle).
	
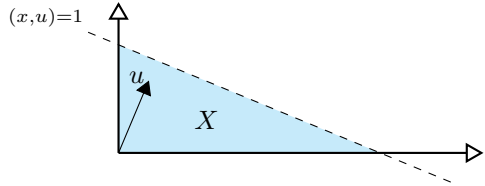
\begin{center}
		\hspace*{-5mm}
		\begin{tikzpicture}[scale=0.4]
		\path [fill=cyan!20] (0,0) -- (0,4-5/12) -- (11-12/5,0) -- (0,0);
		\draw [->,>=open triangle 60,thick] (0,0) -- (12,0);
		\draw [->,>=open triangle 60,thick] (0,0) -- (0,5);
		\draw [dashed] (-1,4) -- (11,-1);
		\draw [->,>=triangle 60] (0,0) -- (1,1*12/5);
		\node at (0.6,2.5) {{$u$}};
		\node at (2.9,1.1) {{$X$}};
		\node at (-2.4,4.5) {{$\scriptstyle(x,u)=1$}};
		\end{tikzpicture}
	\end{center}
	The $n$-simplex $\Delta^n$ is of course a particular pyramid. It is obtained using the vector 
	$u=(1,\ldots,1)$. 

	The Minkowski functional of the pyramid $X$ associated with $u$ is 
	\[
		\Min X(x)=
		\begin{cases}
			(x,u) &\hspace*{-3mm},\quad x\geq 0,
			\\
			\infty &\hspace*{-3mm},\quad \text{otherwise}.
		\end{cases}
	\]
	Write $u=\sum_{j=1}^n\alpha_je_j$, where $e_j$ is the $j$-th canonical basis vector, 
	and set $y_j=\frac 1{\alpha_j}e_j$. Clearly, $\{y_1,\ldots,y_n\}$ is linearly independent. 
	Each vector $x=\sum_{j=1}^n\xi_je_j$ can be written as $x=\sum_{j=1}^n(\xi_j\alpha_j)y_j$, 
	and this is a subconvex combination if and only if $\xi_j\geq 0$ and $\sum_{j=1}^n\xi_j\alpha_j\leq 1$, 
	i.e., if and only if $x\in X$. Thus $X$ is generated by $\{y_1,\ldots,y_n\}$ as a \Pca. 

	The linear map given by the diagonal matrix made up of the $\alpha_j$'s induces a bijection of $X$ 
	onto $\Delta^n$, and maps the $y_j$'s to the corner points of $\Delta^n$. Hence, $X$ is free 
	with basis $\{y_1,\ldots,y_n\}$. 

	If $u$ is not strictly positive, the situation changes drastically. Then $X$ is not finitely generated as a \Pca, 
	because it is unbounded whereas the subconvex hull of a finite set is certainly bounded. 
	
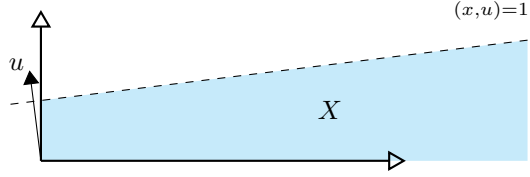
\begin{center}
		\hspace*{-5mm}
		\begin{tikzpicture}[scale=0.4]
		\path [fill=cyan!20] (0,0) -- (0,2) -- (16,4) -- (16,0) -- (0,0);
		\draw [->,>=open triangle 60,thick] (0,0) -- (12,0);
		\draw [->,>=open triangle 60,thick] (0,0) -- (0,5);
		\draw [dashed] (-1,2-1/8) -- (16,4);
		\draw [->,>=triangle 60] (0,0) -- (-3/8,3);
		\node at (-0.8,3.2) {{$u$}};
		\node at (9.5,1.7) {{$X$}};
		\node at (14.8,5.0) {{$\scriptstyle(x,u)=1$}};
		\end{tikzpicture}
	\end{center}
\end{example}

\noindent
Now we return to the functor \Gh. 

\begin{lemma}\label{Ghat-Minko}
	Let $X\subseteq\mathbb R^n$ be a \Pca, and assume that $X$ is compact. Then 
	\[
		\Gh X=\Big\{(o,\phi)\in\mathbb R\times(\mathbb R^n)^A\mid\ 
		o\geq 0,\ o+\sum_{a\in A}\Min X(\phi(a))\leq 1 \Big\}
		.
	\]
\end{lemma}

In the following we use the elementary fact that every convex map has a linear extension.

\begin{lemma}\label{LinExt}
	Let $V_1,V_2$ be vector spaces, let $X\subseteq V_1$ be a \Pca,  and let $c\colon X\to V_2$ be a convex map.
	Then $c$ has a linear extension $\tilde c\colon V_1\to V_2$. If $\spn X=V_1$, this extension is unique.
\end{lemma}

Rescaling in this representation of $\Gh X$ leads to a characterisation of \Gh-coalgebra maps. 
We give a slightly more general statement; for the just said, use $X=Y$.

\begin{corollary}\label{Ghat-Minko-Coalg}
	Let $X,Y\subseteq\mathbb R^n$ be \Pca s, and assume that $X$ and $Y$ are compact. 
	Further, let $c\colon X\to\mathbb R_+\times(\mathbb R^n)^A$ be a convex map, 
	and let $\tilde c\colon \mathbb R^n\to\mathbb R\times(\mathbb R^n)^A$ be a linear extension of $c$.

	Then $c(X)\subseteq\Gh Y$, if and only if 
	\begin{equation}\label{2}
		\Out{\tilde c}(x)+\sum_{a\in A}\Min Y(\Comp{\tilde c}a(x))\leq\Min X(x),\quad x\in\mathbb R^n
		.
	\end{equation}
\end{corollary}

\section{An extension theorem for $\Gh$-coalgebras}
\label{sec-ExtensionTheorem}

In this section we establish an extension theorem for $\Gh$-coalgebras. 
It states that an $\Gh$-coalgebra, whose carrier has a particular geometric form, 
can, under a mild additional condition, be embedded into an $\Gh$-coalgebra 
whose carrier is free and finitely generated.

\begin{theorem}\label{ExEx}
	Let $(X,c)$ be an $\Gh$-coalgebra whose carrier $X$ is a compact subset of 
	a euclidean space $\mathbb R^n$ with $\Delta^n\subseteq X\subseteq\mathbb R_+^n$.
	Assume that the output map $\Out c$ does not vanish on invariant coordinate hyperplanes  
	in the sense that ($e_j$ denotes again the $j$-th canonical basis vector in $\mathbb R^n$)
	\begin{equation}\label{nv}
		\begin{aligned}
		& \nexists\,I\subseteq\{1,\ldots,n\}\Deli
		\\
		& I\neq\emptyset,\quad 
			\Out c(e_j)=0,j\in I,
			\quad
			\Comp ca(e_j)\subseteq\spn\{e_i\mid i\in I\},a\in A,j\in I.
		\end{aligned}
	\end{equation}
	Then there exists an $\Gh$-coalgebra $(Y,d)$, such that $X\subseteq Y\subseteq\mathbb R_+^n$,
	the inclusion map $\iota\colon X\to Y$ is a $\CoCat{\Gh}$-morphism,
	and $Y$ is the subconvex hull of $n$ linearly independent vectors (in particular, $Y$ is free with $n$ generators).
\end{theorem}

The idea of the proof can be explained by geometric intuition. 
Say, we have an $\Gh$-coalgebra $(X,c)$ of the stated form, and 
let $\tilde c\colon \mathbb R^n\to\mathbb R\times(\mathbb R^n)^A$ be the linear extension of $c$ to all of $\mathbb R^n$, 
cf.\ Lemma~\ref{LinExt}.
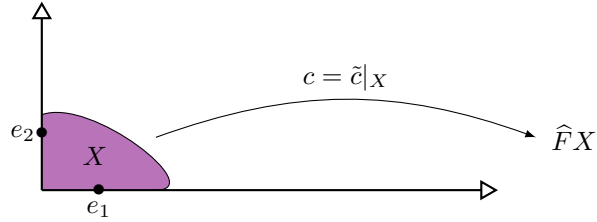
\begin{center}
	\hspace*{-5mm}
	\begin{tikzpicture}[scale=0.5]
	\path [fill=violet!55] (0,0) -- (0,2) to [out=20,in=0] (3,0) -- (0,0);
	\draw [->,>=open triangle 60,thick] (0,0) -- (12,0);
	\draw [->,>=open triangle 60,thick] (0,0) -- (0,5);
	\node at (0,1.5) {{$\bullet$}};
	\node at (-0.5,1.5) {{$e_2$}};
	\node at (1.5,0) {{$\bullet$}};
	\node at (1.5,-0.5) {{$e_1$}};
	\node at (14,1.4) {$\Gh X$};
	\draw (0,2) to [out=20,in=0] (3,0);
	\node at (1.4,0.9) {{$X$}};
	\draw [->,>=latex] (3,1.4) to [out=20,in=160] (13,1.4);
	\node at (8,3) {$c=\tilde c|_X$};
	\end{tikzpicture}
\end{center}
Remembering that pyramids are free and finitely generated, we will be done if we find a pyramid $Y\supseteq X$ which is mapped 
into $\Gh Y$ by $\tilde c$:
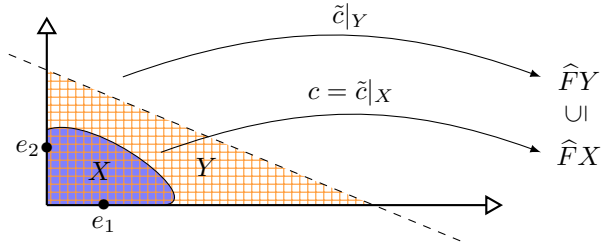
\begin{center}
	\hspace*{-5mm}
	\begin{tikzpicture}[scale=0.5]
	\path [fill=blue!50] (0,0) -- (0,2) to [out=20,in=0] (3,0) -- (0,0);
	\path [pattern color=orange!80,pattern=grid] (0,0) -- (0,4-5/12) -- (11-12/5,0) -- (0,0);
	\draw [->,>=open triangle 60,thick] (0,0) -- (12,0);
	\draw [->,>=open triangle 60,thick] (0,0) -- (0,5);
	\node at (0,1.5) {{$\bullet$}};
	\node at (-0.5,1.5) {{$e_2$}};
	\node at (1.5,0) {{$\bullet$}};
	\node at (1.5,-0.5) {{$e_1$}};
	\node at (14,1.4) {$\Gh X$};
	\draw (0,2) to [out=20,in=0] (3,0);
	\node at (1.4,0.9) {{$X$}};
	\draw [->,>=latex] (3,1.4) to [out=20,in=160] (13,1.4);
	\node at (8,3) {$c=\tilde c|_X$};
	\draw [dashed] (-1,4) -- (11,-1);
	\node at (4.2,1) {{$Y$}};
	\node at (14,3.4) {$\Gh Y$};
	\node at (14.1,2.15) {\begin{rotate}{90} $\subseteq$ \end{rotate}};
	\draw [->,>=latex] (2,3.4) to [out=20,in=160] (13,3.4);
	\node at (8,5) {$\tilde c|_Y$};
	\end{tikzpicture}
\end{center}
This task can be reformulated as follows:
For each pyramid $Y_1$ containing $X$ let $P(Y_1)$ be the set of all pyramids $Y_2$ containing $X$, 
such that $\tilde c(Y_2)\subseteq\Gh Y_1$. If we find $Y$ with $Y\in P(Y)$, we are done.

Existence of $Y$ can be established by applying a fixed point principle for set-valued maps. 
The result sufficient for our present level of generality is Kakutani's generalisation \cite[Corollary]{kakutani:1941} 
of Brouwers fixed point theorem.

\section{Properness of \Gh}
\label{sec-PropernessSubcubic}

In this section we give the second main result of the paper.

\begin{theorem}\label{GhatProp}
	The functor $\Gh$ is proper. 

	In fact, for each two given coalgebras with free finitely generated carrier and each two elements having 
	the same trace, a zig-zag with free and finitely generated nodes relating those elements can be found, 
	which has three intermediate nodes with the middle one forming a span.
\end{theorem}

We try to follow the proof scheme familiar from the cubic case. 
Assume we are given two $\Gh$-coalgebras with free finitely generated carrier, say $(\Delta^{n_1},c_1)$ and $(\Delta^{n_2},c_2)$, 
and elements $x_1\in\Delta^{n_1}$ and $x_2\in\Delta^{n_2}$ having the same trace. 
Since $\Gh\Delta^{n_j}\subseteq\mathbb R\times(\mathbb R^{n_j})^A$ we can apply Lemma~\ref{LinExt}
and obtain $\Cub{\mathbb R}$-coalgebras $(\mathbb R^{n_j},\tilde c_j)$ with $\tilde c_j|_{\Delta^{n_j}}=c_j$. 
This leads to the basic diagram:
\[
\xymatrix@C=15pt@R=35pt@M=7pt{
	\Delta^{n_1} \ar@{.>}[r] \ar@{.>}[d]_{c_1} & \mathbb R^{n_1} \ar[d]_{\tilde c_1}
	&& Z \ar[ll]_{\pi_1} \ar[rr]^{\pi_2} \ar[d]^d 
		\save[]+<3pt,10pt>*\txt{\begin{rotate}{90}$\subseteq$\end{rotate}}\restore
		\save[]+<2pt,27pt>*\txt{$\mathbb R^{n_1}\!\!\times\mathbb R^{n_2}$}\restore
	&& \mathbb R^{n_2} \ar[d]_{\tilde c_2} & \Delta^{n_2} \ar@{.>}[l] \ar@{.>}[d]^{c_2}
	\\
	\Gh \Delta^{n_1} \ar@{.>}[r] & \Cub{\mathbb R}\mathbb R^{n_1} 
	&& \Cub{\mathbb R}Z \ar[ll]_{\Cub{\mathbb R}\pi_1} \ar[rr]^{\Cub{\mathbb R}\pi_2}
		\save[]+<-3pt,-10pt>*\txt{\begin{rotate}{-90}$\subseteq$\end{rotate}}\restore
		\save[]+<2pt,-29pt>*\txt{$\mathbb R\!\!\times(\mathbb R^{n_1}\!\!\times\mathbb R^{n_2})^A$}\restore
	&& \Cub{\mathbb R}\mathbb R^{n_2} & \Gh \Delta^{n_2} \ar@{.>}[l]
}
\]
At this point the line of argument known from the cubic case breaks: 
it is \emph{not} granted that $Z\cap(\Delta^{n_1}\times \Delta^{n_2})$ becomes an $\Gh$-coalgebra with the restriction of $d$. 

The substitute for $Z\cap(\Delta^{n_1}\times \Delta^{n_2})$ suitable for proceeding one step further is given 
by the following lemma, where we tacitly identify $\mathbb R^{n_1}\times\mathbb R^{n_2}$ with $\mathbb R^{n_1+n_2}$. 

\begin{lemma}\label{GhatRestr}
	We have $d(Z\cap 2\Delta^{n_1+n_2})\subseteq\Gh(Z\cap 2\Delta^{n_1+n_2})$.
\end{lemma}

This shows that $Z\cap 2\Delta^{n_1+n_2}$ becomes an $\Gh$-coalgebra with the restriction of $d$. 
Still, we cannot return to the usual line of argument: it is \emph{not} granted that 
$\pi_j(Z\cap 2\Delta^{n_1+n_2})\subseteq\Delta^{n_j}$. This forces us to introduce additional nodes to 
produce a zig-zag in $\CoCat{\Gh}$. These additional nodes are given by the following lemma.
There $\co( -)$ denotes the convex hull. 

\begin{lemma}\label{GhatNodes}
	Set $Y_j=\co(\Delta^{n_j}\cup\pi_j(Z\cap 2\Delta^{n_1+n_2}))$. Then $\tilde c_j(Y_j)\subseteq\Gh Y_j$.
\end{lemma}

This shows that $Y_j$ becomes an $\Gh$-coalgebra with the restriction of $\tilde c_j$. 
We are led to a zig-zag in $\CoCat{\Gh}$:
\[
\xymatrix@C=25pt{
	(\Delta^{n_1},c_1) \ar[r]^{\subseteq} 
	& (Y_1,\tilde c_1) & \big(Z\cap 2\Delta^{n_1+n_2},d\big) \ar[l]_{\pi_1\mkern40mu} \ar[r]^{\mkern40mu\pi_2} 
	& (Y_2,\tilde c_2) & (\Delta^{n_2},c_2) \ar[l]_{\supseteq}
}
\]
This zig-zag relates $x_1$ and $x_2$ since $(x_1,x_2)\in Z\cap 2\Delta^{n_1+n_2}$. 

Using Minkowski's Theorem and the argument from Lemma~\ref{RedLemPca} (Appendix~\ref{app:B}) shows that 
the middle node has finitely generated carrier. The two nodes with incoming arrows 
are, as convex hulls of two finitely generated \Pca s, of course also finitely generated. 
But in general they will not be free (and this is essential, remember Remark~\ref{fg}). 
Now Theorem~\ref{ExEx} comes into play. 

\begin{lemma}\label{GhatFree}
	Assume that each of $(\Delta^{n_1},c_1)$ and $(\Delta^{n_2},c_2)$ satisfies the following condition:
	\begin{equation}\label{nv2}
		\begin{aligned}
		& \nexists\,I\subseteq\{1,\ldots,n\}\Deli
		\\
		& I\neq\emptyset,\ 
			\Out{c_j}(e_k)=0,k\in I,
			\ 
			\Comp{c_j}a(e_k)\subseteq\co(\{e_i\mid i\in I\}\cup\{0\}),a\in A,k\in I.
		\end{aligned}
	\end{equation}
	Then there exist free finitely generated \Pca s $U_j$ with $Y_j\subseteq U_j\subseteq\mathbb R_+^{n_j}$ which satisfy 
	$\tilde c_j(U_j)\subseteq\Gh U_j$.
\end{lemma}

This shows that $U_j$, under the additional assumption \eqref{nv2} on $(\Delta^{n_j},c_j)$, 
becomes an $\Gh$-coalgebra with the restriction of $\tilde c_j$. 
Thus we have a zig-zag in $\CoCat{\Gh}$ relating $x_1$ and $x_2$ whose nodes with incoming arrows are free and finitely generated, 
and whose node with outgoing arrows is finitely generated:
\[
\xymatrix@C=25pt{
	(\Delta^{n_1},c_1) \ar@{.>}[r]^{\subseteq} \ar[rd]
	& (Y_1,\tilde c_1) \ar@{.>}[d]
		\save[]+<6pt,-15pt>*\txt{\begin{rotate}{-90}$\scriptstyle\subseteq$\end{rotate}}\restore
	& \big(Z\cap 2\Delta^{n_1+n_2},d\big) \ar@{.>}[l]_{\pi_1\mkern40mu} \ar@{.>}[r]^{\mkern40mu\pi_2} \ar[ld] \ar[rd]
	& (Y_2,\tilde c_2) \ar@{.>}[d]
		\save[]+<-9pt,-15pt>*\txt{\begin{rotate}{-90}$\scriptstyle\subseteq$\end{rotate}}\restore
	& (\Delta^{n_2},c_2) \ar@{.>}[l]_{\supseteq} \ar[ld]
	\\
	& (U_1,\tilde c_1)
	&& (U_2,\tilde c_2) &
}
\]
Removing the additional assumption on $(\Delta^{n_j},c_j)$ is an easy exercise.

\begin{lemma}\label{GhatRedLem}
	Let $(\Delta^n,c)$ be an $\Gh$-coalgebra. 
	Assume that $I$ is a nonempty subset of $\{1,\ldots,n\}$ with 
	\begin{equation}\label{6}
		\Out c(e_k)=0,\ k\in I\quad\text{and}\quad 
		\Comp ca(e_k)\in\co\big(\{e_i\mid i\in I\}\cup\{0\}\big),\ a\in A,k\in I.
	\end{equation}
	Let $X$ be the free \Pca\ with basis $\{e_k\mid k\in\{1,\ldots,n\}\setminus I\}$, and let $f\colon \Delta^n\to X$ 
	be the \Pca-morphism with 
	\[
		f(e_k)=
		\begin{cases}
			0 &\hspace*{-3mm},\quad k\in I,
			\\
			e_k &\hspace*{-3mm},\quad k\not\in I.
		\end{cases}
	\]
	Further, let $g\colon X\to[0,1]\times X^A$ be the \Pca-morphism with 
	\[
		g(e_k)=\Big(\Out c(e_k),\big(f(\Comp ca(e_k))\big)_{a\in A}\Big),\quad k\in\{1,\ldots,n\}\setminus I
		.
	\]
	Then $(X,g)$ is an $\Gh$-coalgebra, and $f$ is an $\Gh$-coalgebra morphism of $(\Delta^n,c)$ onto 
	$(X,g)$. 
\end{lemma}

\begin{corollary}\label{GhatRedCor}
	Let $(\Delta^n,c)$ be an $\Gh$-coalgebra. Then there exists $k\leq n$, an $\Gh$-coalgebra $(\Delta^k,g)$, 
	such that $(\Delta^k,g)$ satisfies the assumption in Lemma~\ref{GhatFree} and such that there exists an 
	$\Gh$-coalgebra map $f$ of $(\Delta^n,c)$ onto $(\Delta^k,g)$. 
\end{corollary}

The proof of Theorem~\ref{GhatProp} is now finished by putting together what we showed so far.
Starting with $\Gh$-coalgebras $(\Delta^{n_j},c_j)$ without any additional assumptions, and elements $x_j\in\Delta^{n_j}$ 
having the same trace, we first reduce by means of Corollary~\ref{GhatRedCor} and then apply Lemma~\ref{GhatFree}. 
This gives a zig-zag as required:
\[
\xymatrix@C=17pt@M=3pt{
	(\Delta^{n_1},c_1) \ar@{.>}[d]_{\psi_1} \ar[rd]
	&& \big(Z\cap 2\Delta^{k_1+k_2},d\big) \ar[ld] \ar[rd] &&
	(\Delta^{n_2},c_2) \ar@{.>}[d]_{\psi_2} \ar[ld]
	\\
	(\Delta^{k_1},g_1) \ar@{.>}[r]
	& (U_1,\tilde g_1) & 
	& (U_2,\tilde g_2) & (\Delta^{k_2},g_2) \ar@{.>}[l]
}
\]
and completes the proof of properness of $\Gh$.


%
%

%
%

\def\cprime{$'$} \def\cprime{$'$} \def\cprime{$'$} \def\cprime{$'$}
  \def\cprime{$'$} \def\cprime{$'$} \def\cprime{$'$}
  \def\polhk#1{\setbox0=\hbox{#1}{\ooalign{\hidewidth
  \lower1.5ex\hbox{`}\hidewidth\crcr\unhbox0}}} \def\cprime{$'$}
  \def\cprime{$'$} \def\cprime{$'$} \def\cprime{$'$} \def\cprime{$'$}
  \def\cprime{$'$} \def\cprime{$'$} \def\cprime{$'$} \def\cprime{$'$}
  \def\cprime{$'$} \def\cprime{$'$} \def\cprime{$'$} \def\cprime{$'$}
  \def\cprime{$'$} \def\cprime{$'$}


\newpage
\appendix
\makeatletter
\renewcommand{\@seccntformat}[1]{\@nameuse{the#1}.\quad}
\makeatother
\renewcommand{\thelemma}{\Alph{section}.\arabic{lemma}}
\section{Category theory basics}
\label{sec:app-basics}
%

%
We start by recalling the basic notions of category, functor and natural transformation, so that all of the results in the paper are accessible also to non-experts.

A category $\cat C$ is a collection of objects and a collection of arrows (or morphisms) from one object to another. For every object $X \in \cat C$, there is an identity arrow $\id_X\colon X \to X$. For any three objects $X, Y, Z \in \cat C$, given two arrows $f\colon X \to Y$ and $g\colon Y\to Z$, there exists an arrow $g \after f\colon X \to Z$. Arrow composition is associative and $\id_X$ is neutral w.r.t. composition. 
The standard example is $\Sets$, the category of sets and functions. 

A functor $F$ from a category $\cat C$ to a category $\cat D$, notation $F\colon \cat C \to \cat D$, assigns to every object $X\in \cat C$, an object $FX \in \cat D$, and to every arrow $f\colon X \to Y$ in $\cat C$ an arrow $Ff\colon FX \to FY$ in $\cat D$ such that identity arrows and composition are preserved.

A concrete category is a category $\cat C$ equipped with a faithful functor $\mathcal U\colon \cat C \to \Sets$. Intuitively, a concrete category has objects that are sets with some additional structure, e.g. algebras, and morphisms that are particular kind of functions, and $\mathcal U$ is a canonical forgetful functor. All categories that we consider are algebraic and hence concrete. 

\begin{wrapfigure}{r}{.1\textwidth}\vspace*{-.95cm}
\xymatrix@R-.5pc{
FX\ar[d]^-{\sigma_X}\ar[r]^-{Ff} & FY\ar[d]^-{\sigma_Y}\\
GX\ar[r]^-{Gf} & GY
}\vspace*{-.55cm}
\end{wrapfigure}

Let $F\colon \cat C \to \cat D$ and $G\colon \cat C \to \cat D$ be two functors. A natural transformation $\sigma\colon F \Rightarrow G$ is a family of arrows $\sigma_X\colon FX \to GX$ in $\cat D$ such that the diagram on the right commutes for all arrows $f\colon X \to Y$. 

\subsection{Monads and Algebras}\label{sec-app:monads}

A monad is a functor $T\colon\cat{C} \rightarrow \cat{C}$ together with two
natural transformations: a unit $\eta\colon \id_{\cat{C}}
\Rightarrow T$ and multiplication $\mu \colon T^{2} \Rightarrow
T$. These are required to make the following diagrams commute, for
$X\in\cat{C}$.
$$\xymatrix@C-.5pc@R-.5pc{
T X\ar[rr]^-{\eta_{T X}}\ar@{=}[drr] & & T^{2}X\ar[d]^{\mu_X} & &
   T X\ar[ll]_{T\eta_{X}}\ar@{=}[dll] 
&  & T^{3}X\ar[rr]^-{\mu_{T X}}\ar[d]_{T\mu_{X}} 
   & & T^{2}X\ar[d]^{\mu_X} \\
& & T X & &
& &
T^{2}X\ar[rr]_{\mu_X} & & T X
}$$
Given two monads $S,T$ with units and multiplications $\eta^S,\eta^T$ and $\mu^S,\mu^T$, respectively, and 
a natural transformation $\iota\colon S\Rightarrow T$, we say that $\iota$ is a monad morphism, and $S\leq T$ along $\iota$, 
if $\eta^T=\iota\circ\eta^S$ and $\iota\circ\mu^S=\mu^T\circ\iota\iota$
where $\iota\iota \stackrel{\text{def}}{=}T\iota\after \iota \stackrel{\text{nat.}}{=} \iota \after S\iota$.

We briefly describe some examples of monads on $\Sets$.
\begin{itemize}
\item The finitely supported subprobability distribution monad $\Dis$ is defined, for a set $X$ and a function $f\colon X \to Y$,  as
$$\Dis X
\,\,\, = \,\,\,
\{\varphi\colon X \to [0,1] \mid \sum_{x \in X} \varphi(x) \le 1,\, \supp(\varphi) \text{~is~finite}\}$$
and $$\Dis f(\varphi)(y)
\,\,\, = \,\,\,
\sum\limits_{x\in f^{-1}(\{y\})} \varphi(x).$$ 
Here and below $\supp(\varphi) = \{x \in X \mid \varphi(x) \neq 0\}$.
The unit of $\Dis$ is given by a Dirac
distribution $\eta_X(x) = \delta_x = ( x \mapsto 1)$ for $x \in X$ and
the multiplication by $\mu_X(\Phi)(x) = \sum\limits_{\varphi \in \supp(\Phi)}
\Phi(\varphi)\cdot \varphi(x)$ for $\Phi \in \Dis\Dis X$. 
\item For a semiring $\mathbb S$ the $\mathbb S$-valuations monad $\Ts$ is defined as 
	$\Ts X = \{\varphi\colon X \to \mathbb S \mid \supp(\varphi) \text{ is finite}\}$ and on functions 
	$f\colon X \to Y$ we have $\Ts f(\varphi)(y) = \sum_{x \in f^{-1}(\{y\})} \varphi(x)$. 
	Its unit is given by $\eta_X(x) = ( x \mapsto 1)$ and multiplication by 
	$\mu_X(\Phi)(x) = \sum_{\varphi \in \supp \Phi} \Phi(\varphi) \cdot \varphi(x)$ for $\Phi \in \Ts\Ts X$.
\item To illustrate the connection between $\Dis$ and $\Ts$, consider yet another monad: For a semiring $\mathbb S$, and 
	a (suitable) subset $S \subseteq \mathbb S$, the ($\mathbb S, S$)-valuations monad $T_{\mathbb S, S}$ is defined as follows. 
	On objects it acts like 
	$$T_{\mathbb S,S}X = \{\varphi\colon X \to \mathbb S \mid \supp(\varphi) \text{ is finite and } 
	\sum_{x \in X} \varphi(x) \in S\}$$ 
	on functions it acts like $\Ts$. The unit and multiplication are defined as in $\Ts$. 
	Note that $\Dis = T_{\mathbb R_+,[0,1]}$. 
\end{itemize}

With a monad $T$ on a category $\cat{C}$ one associates the  Eilenberg-Moore category
$\cat{C}^{T}$ of Eilenberg-Moore algebras. Objects of
$\cat{C}^{T}$ are pairs $\alg A = (A, \alpha)$ of an object $A \in \cat C$ and an arrow
$\alpha\colon T A \rightarrow A$, making the first two
diagrams below commute. 
$$\xymatrix@R-.5pc{
A\ar@{=}[dr]\ar[r]^-{\eta_A} & T A\ar[d]^{\alpha}
& 
T^{2}A\ar[d]_{\mu_A}\ar[r]^-{T\alpha} & T A\ar[d]^{\alpha}
& &
T A\ar[d]_{\alpha}\ar[r]^-{T h} & T B\ar[d]^{\beta} \\
& A 
&
T A\ar[r]_-{\alpha} & A
& &
A\ar[r]_-{h} & B
}$$

\noindent A homomorphism from an algebra $\alg A = (A, \alpha)$ to an algebra $\alg B = (B, b)$ is a map $h\colon
A\rightarrow B$ in $\cat{C}$ between the underlying objects making the
diagram above on the right commute.  

From now on fix $\cat{C}=\Sets$.
A free Eilenberg-Moore algebra for a monad $T$ generated by $X$ is $(TX, \mu_X)$ and we will often denote it simply by $TX$. 
A free finitely generated Eilenberg-Moore algebra for $T$ is an algebra $TX$ with $X$ a finite set. The diagram in the middle thus
says that the map $\alpha$ is a homomorphism from $T A$ to $\alg A$.

Indeed, $(TX,\mu_X)$ is free on $X$ as for every $T$-algebra $\alg A = (A, \alpha)$ and any $\Sets$-morphism $f\colon X \to A$ there is a unique $\EM{T}$-morphism $f^\#\colon (TX,\mu_X) \to \mathbb A$ such that $f^\# \after \eta_A = f$ --- it is easy to see that this unique extension $f^\#$ is the Kleisli extension of $f$, i.e., $f^\# = \alpha \after Tf$.  Moreover, note that $(f^\# \after \eta_A)^\# = f^\#$, by the uniqueness of the extension.

\section{Proof details for properness of cubic functors}\label{app:B}

\begin{Proof}[of Lemma~\ref{Prod}]
	Since $\tr_{c_1} x_1=\tr_{c_2} x_2$, we have 
	\[
		\Out{c_1}(\Comp{c_1}w(x_1))=[\tr_{c_1} x_1](w)=[\tr_{c_2} x_2](w)=\Out{c_2}(\Comp{c_2}w(x_2)),\quad w\in A^*,
	\]
	and therefore $d_1|_Z=d_2|_Z$. Moreover, 
	\[
		\Comp{c_j}a(\Comp{c_j}w(x_j))=\Comp{c_j}{wa}(x_j),\quad w\in A^*,
	\]
	and therefore $d_j(Z)\subseteq\mathbb S\times Z^A$.
\end{Proof}

\begin{Proof}[of Corollary~\ref{Noetherian}]
	Remembering Remark~\ref{ProperSemiring}, we have to show that the functor $\Cub{\mathbb S}$ is proper. 
	We have the zig-zag \eqref{ZZ}, and the $\mathbb S$-semimodule $Z$ is, 
	as a subsemimodule of the finitely generated $\mathbb S$-semimodule 
	$\mathbb S^{n_1}\times\mathbb S^{n_2}$, itself finitely generated. 
\end{Proof}

\subsection{Proof of the extension lemma}

The proof of the extension lemma follows directly from the following two abstract properties.

\begin{lemma}\label{lem:ext-lem-1}
	Assume $\Ts \le \Te$ via $\iota\colon \Ts \Rightarrow \Te$ and let $X$ be a finite set. Let $\alg{Y} \in \EM{\Ts}$ and $\alg{Z} \in \EM{\Te}$ and assume we are given an arrow $a_Y\colon \Ts X \to \alg{Y}$ in $\EM{\Ts}$ and a $\Ts \le \Te$-homomorphism  $h\colon \alg{Y} \to \alg{Z}$. Then there exists an arrow $a_Z\colon \Te X \to \alg{Z}$ in $\EM{\Te}$ making the following diagram commute. 
		$$\xymatrix@R=0.7em{
	\Ts X \ar[rr]^{\iota} \ar[d]_{a_Y}&& \Te X \ar[d]^{a_Z}\\
	\alg{Y} \ar[rr]^{h}&& \alg{Z}
	}$$
\end{lemma}

\begin{proof}
	Consider the map $h \after a_Y \after \eta_{\mathbb S,X}\colon X \to Z$. Let $a_Z =(h \after a_Y \after \eta_{\mathbb S,X})^{\#_{\mathbb{E}}}$. 

All morphisms in the square are $\EM{\Ts}$-morphisms: $\alpha_Y$ by definition; $\iota_x$ as one of the monad morphism laws shows this; $h$ as it is a $\Ts \le \Te$-homomorphism; and $\alpha_Z = M_\iota(\alpha_Z)$. Clearly, then $h \after \alpha_Y$ and $\alpha_Z \after \iota_X$ are $\EM{\Ts}$-morphisms from the free algebra $(\Ts X, \mu_X)$ to $\alg Z$. 

Therefore, for the commutativity of the square it suffices to show that $$h \after \alpha_Y \after \eta_{\mathbb S} = \alpha_Z \after \iota_X \after \eta_{\mathbb S}$$ as then, by the uniqueness of the extension, $$h \after \alpha_Y = (h\after \alpha_Y \after \eta_{\mathbb S})^{\#_{\mathbb{S}}} = 
(\alpha_Z \after \iota_X \after \eta_{\mathbb S})^{\#_{\mathbb{S}}} = \alpha_Z \after \iota_X.$$

The last needed equality follows because $\Ts \le \Te$ along $\iota$ and so
$$(h \after a_Y \after \eta_{\mathbb S,X})^{\#_{\mathbb{E}}} \after \iota_X \after \eta_{\mathbb{S},X} = (h \after a_Y \after \eta_{\mathbb S,X})^{\#_{\mathbb{E}}} \after \eta_{\mathbb E,X} = h \after a_Y \after \eta_{\mathbb S,X}. \vspace*{-7mm}$$\qed 
\end{proof}

\begin{lemma}\label{lem:ext-lem-2}
The map $\iota \times \iota^A$ is a $\Ts\le\Te$-homomorphism from $\Cub{\mathbb{S}}(\Ts X)$ to $\Cub{\mathbb{E}}(\Te X)$.
\end{lemma}

\begin{proof}
Since the functor $M_\iota$ induced by the monad morphism $\iota$ satisfies $\mathcal U_\mathbb{S} \after M_\iota = \mathcal U_\mathbb{E}$, it preserves all limits (as $\mathcal U_\mathbb{E}$ preserves them and $\mathcal U_\mathbb{S}$ reflects them). Hence, in particular, it preserves products.
Since $\iota_X\colon (\Ts X, \mu_{\mathbb S,X}) \to M_\iota(\Te X, \mu_{\mathbb E,X})$ is a $\Ts$-algebra homomorphism by one of the monad morphism laws, we have 
$$\iota_1 \times \iota_X^A \colon \Ts 1 \times \Ts X^A \to M_\iota(\Te 1) \times M_\iota(\Te X)^A = M_\iota(\Te 1 \times \Te X^A)$$
is one as well. \qed
\end{proof}

\begin{Proof}[of Lemma~\ref{CubProperties}]
	Since $\mathbb S\subseteq\mathbb E$, we have 
	\[
		\Cub{\mathbb E}X\cap\Cub{\mathbb S}Y=(\mathbb E\times X^A)\cap(\mathbb S\times Y^A)
		=\mathbb S\times(X\cap Y)^A=\Cub{\mathbb S}(X\cap Y)
		.
	\]
	Assume now that $Y_j\subseteq X_j$. We have 
	\[
		(\Cub{\mathbb E}\pi_1)^{-1}(\Cub{\mathbb S}Y_1)=
		\{(o,(({x_1}_a,{x_2}_a))_{a\in A})\in\mathbb E\times(X_1\times X_2)^A\mid o\in\mathbb S,{x_1}_a\in Y_1\}
		,
	\]
	and the analogous formula for $(\Cub{\mathbb E}\pi_2)^{-1}(\Cub{\mathbb S}Y_2)$. 
	This shows that the intersection of these two inverse images is equal to $\mathbb S\times(Y_1\times Y_2)^A$.
\end{Proof}

\subsection{Proof of the reduction lemma}

\subsection*{\ding{226}\ Reducing from \Ab\ to \CMon}

The reduction lemma for passing from abelian groups to commutative monoids arises from a classical result of algebra. 
Namely, it is a corollary of the following theorem due to D.Hilbert, 
cf.\ \cite[Theorem~II]{hilbert:1890} see also \cite[Theorem~1.1]{bachem:1978}.

\begin{theorem}[Hilbert 1890]\label{PolyhedralZ}
	Let $W$ be a $n\times m$-matrix with integer entries, and let $X$ be the commutative monoid 
	\[
		X=\big\{x\in\mathbb Z^n\mid x\cdot W\geq 0\big\}
		,
	\]
	where the monoid operation is the usual addition on $\mathbb Z^n$. 
	Then $X$ is finitely generated as a commutative monoid. 
\end{theorem}

The reduction lemma for passing from \Ab\ to \CMon\ is a corollary
Since every finitely generated abelian group is also finitely generated as a commutative monoid, 
we obtain a somewhat stronger variant.

\begin{lemma}\label{RedLemN}
	Let $Z$ be a finitely generated abelian group, let $m\in\mathbb N$, and let 
	$\varphi\colon Z\to\mathbb Z^m$ be a group homomorphism. 
	Then $\varphi^{-1}(\mathbb N^m)$ is finitely generated as a commutative monoid. 
\end{lemma}
\begin{Proof}
	Write $Z$, up to an isomorphism, as a direct sum of cyclic abelian groups
	\begin{equation}\label{3}
		Z=\mathbb Z^k\oplus\Big[\bigoplus_{j=1}^n\mathbb Z/a_j\mathbb Z\Big]
	\end{equation}
	with $a_j\geq 2$. Since $\varphi$ maps into the torsionfree group $\mathbb Z^m$, we must have 
	\[
		\varphi\Big(\bigoplus_{j=1}^n\mathbb Z/a_j\mathbb Z\Big)=\{0\}
		.
	\]
	Hence, an element $x\in Z$ satisfies $\varphi(x)\geq 0$, if and only if $\varphi(x_0)\geq 0$ where 
	$x=x_0+x_1$ is the decomposition of $x$ according to the direct sum \eqref{3}.
	The action of the map $\psi=\varphi|_{\mathbb Z^k}\colon \mathbb Z^k\to\mathbb Z^m$ is described as multiplication of 
	$x_0=(\xi_1,\ldots,\xi_k)$ with some $k\times m$-matrix $W$ having integer coefficients. Thus
	\[
		\psi^{-1}(\mathbb N^m)=\big\{x_0\in\mathbb Z^k\mid x_0\cdot W\geq 0\big\}
		,
	\]
	and by Hilbert's Theorem $\psi^{-1}(\mathbb N^m)$ is finitely generated as a commutative monoid. 

	The set $\bigoplus_{j=1}^n\mathbb Z/a_j\mathbb Z$ also has a finite set of generators as a monoid, for example 
	the residue classes $1/a_j\mathbb Z$, $j=1,\ldots,n$. Together we see that $\varphi^{-1}(\mathbb N^m)$ 
	has a finite set of generators as a commutative monoid.
\end{Proof}

\subsection*{\ding{226}\ Reducing from \Vect Q\ to \Mod{Q_+}}

The reduction lemma for passing from vector spaces over $\mathbb Q$ to $\mathbb Q_+$-semimodules 
is a corollary of the one passing from \Ab\ to \CMon. Thus we have the corresponding stronger variant also in this case.

\begin{lemma}\label{RedLemQ+}
	Let $Z$ be a finite dimensional $\mathbb Q$-vector space, let $m\in\mathbb N$, and let 
	$\varphi\colon Z\to\mathbb Q^m$ be $\mathbb Q$-linear. 
	Then $\varphi^{-1}(\mathbb Q_+^m)$ is finitely generated as a $\mathbb Q_+$-semimodule.
\end{lemma}
\begin{Proof}
	Let $\{u_1,\ldots,u_k\}$ be a set of generators of $Z$ as a $\mathbb Q$-vector space.
	Write 
	\[
		\varphi(u_j)=\Big(\frac{a_{j,1}}{b_{j,1}},\ldots,\frac{a_{j,m}}{b_{j,m}}\Big),\quad 
		j=1,\ldots,k
		,
	\]
	with $a_{j,i}\in\mathbb Z$ and $b_{j,i}\in\mathbb N\setminus\{0\}$. 
	Set $b=\prod_{j=1}^k\prod_{i=1}^mb_{j,i}$, then $\varphi(bu_j)\in\mathbb Z^m$, $j=1,\ldots,k$. 

	Let $Z'\subseteq Z$ be the $\mathbb Z$-submodule generated by $\{bu_1,\ldots,bu_k\}$, and set $\psi=\varphi|_{Z'}$. 
	Then $\psi$ is a $\mathbb Z$-linear map of $Z'$ into $\mathbb Z^m$. 
	By Lemma~\ref{RedLemN}, $\psi^{-1}(\mathbb N^m)$ is finitely generated as $\mathbb N$-semimodule, 
	say by $\{v_1,\ldots,v_l\}\subseteq Z'$.

	Given $x\in\varphi^{-1}(\mathbb Q_+^m)$, choose $\nu_1,\ldots,\nu_k\in\mathbb Q$ with 
	$x=\sum_{j=1}^k\nu_ju_j$. Write $\nu_j=\frac{\alpha_j}{\beta_j}$ with $\alpha_j\in\mathbb Z$ and 
	$\beta_j\in\mathbb N\setminus\{0\}$, and set $\beta=\prod_{j=1}^k\beta_j$. Then 
	\[
		\beta b\cdot x=\sum_{j=1}^k(\beta\nu_j)\cdot bu_j\in Z'
		,
	\]
	and 
	\[
		\psi(\beta b\cdot x)=\varphi(\beta b\cdot x)=\beta b\cdot\varphi(x)\in\mathbb Q_+^m\cap\mathbb Z^m=\mathbb N^m
		.
	\]
	Thus $\beta b\cdot x$ is an $\mathbb N$-linear combination of the elements $v_1,\ldots,v_l$, and hence 
	$x$ is a $\mathbb Q_+$-linear combination of these elements. 
	This shows that $\varphi^{-1}(\mathbb Q_+^m)$ is generated by $\{v_1,\ldots,v_l\}$ as a $\mathbb Q_+$-semimodule. 
\end{Proof}

\subsection*{\ding{226}\ Reducing from \Vect R\ to \Cone}

The reduction lemma for passing from vector spaces over $\mathbb R$ to convex cones arises from a different source than 
the previously studied. 
Namely, it is a corollary of the below classical theorem of H.Minkowski, 
cf.\ \cite{minkowski:1896} see also \cite[Theorem~19.1]{rockafellar:1970}.

Recall that a convex subset $X$ of $\mathbb R^n$ is called \emph{polyhedral}, if it is a finite intersection of half-spaces, 
i.e., if there exist $l\in\mathbb N$, $u_1,\ldots,u_l\in\mathbb R^n$, and $\nu_1,\ldots,\nu_l\in\mathbb R$, such that 
\[
	X=\big\{x\in\mathbb R^n\mid (x,u_j)\leq\nu_j,j=1,\ldots,l\big\}
	,
\]
where $(\cdot,\cdot)$ denotes the euclidean scalar product on $\mathbb R^n$. 
On the other hand, $X$ is said to be \emph{generated by points $a_1,\ldots,a_{l_1}$ and directions $b_1,\ldots,b_{l_2}$}, if 
\[
	X=\Big\{\sum_{j=1}^{l_1}\alpha_ja_j+\sum_{j=1}^{l_2}\beta_jb_j\mid 
	\alpha_j\in[0,1],\sum_{j=1}^{l_1}\alpha_j=1,\ \beta_j\geq 0,j=1,\ldots,l_2\Big\}
	.
\]
Note that a convex set generated by some points and directions is bounded, if and only if no (nonzero) directions are present.
Further, a convex set is a cone, if and only if it allows a representation where only directions occur.

\begin{theorem}[Minkowski 1896]\label{PolyhedralR}
	Let $X$ be a convex subset of $\mathbb R^n$. Then $X$ is polyhedral, if and only if $X$ is generated by a finite 
	set of points and directions.
\end{theorem}

The relevance of Minkowski's Theorem in the present context is that it shows that the intersection of two finitely generated 
sets is finitely generated (since the intersection of two polyhedral sets is obviously polyhedral). 

The reduction lemma for passing from \Vect R\ to \Cone\ is an immediate corollary.
Since every finite dimensional $\mathbb R$-vector space is also finitely generated as a convex cone, we have the 
corresponding stronger version.

\begin{lemma}\label{RedLemR+}
	Let $Z$ be a finite dimensional $\mathbb R$-vector space, let $m\in\mathbb N$, and let 
	$\varphi\colon Z\to\mathbb R^m$ be $\mathbb R$-linear. 
	Then $\varphi^{-1}(\mathbb R_+^m)$ is finitely generated as a convex cone. 
\end{lemma}
\begin{Proof}
	\hspace*{0pt}
	\begin{Nili}
	\item \emph{Step~1:}\ 
		The image $\varphi(Z)$ is a linear subspace of $\mathbb R^m$, in particular, polyhedral. The positive 
		cone $\mathbb R_+^m$ is obviously also polyhedral. We conclude that the convex cone 
		$\varphi(Z)\cap\mathbb R_+^m$ is generated by some finite set of directions.
	\item \emph{Step~2:}\ 
		The kernel $\varphi^{-1}(\{0\})$ is, as a linear subspace of the finite dimensional vector space $Z$, 
		itself finite dimensional (generated, say, by $\{u_1,\ldots,u_k\}$). 
		Thus it is also finitely generated as a convex cone (in fact, $\{\pm u_1,\ldots,\pm u_k\}$ is a set of generators).

		Choose a finite set of directions $\{a_1,\ldots,a_l\}$ generating $\varphi(Z)\cap\mathbb R_+^m$ as a 
		convex cone, and choose $v_j\in Z$ with $\varphi(v_j)=a_j$, $j=1,\ldots,l$. 
		We claim that $\{\pm u_1,\ldots,\pm u_k\}\cup\{v_1,\ldots,v_l\}$ generates $\varphi^{-1}(\mathbb R_+^m)$ 
		as a convex cone. 
		To see this, let $x\in \varphi^{-1}(\mathbb R_+^m)$. 
		Choose $\alpha_1,\ldots,\alpha_l\geq 0$ with $\varphi(x)=\alpha_1a_1+\ldots+\alpha_la_l$. 
		Then 
		\[
			\varphi\big(x-(\alpha_1v_1+\ldots+\alpha_lv_l)\big)=
			\varphi(x)-\big(\alpha_1\varphi(v_1)+\ldots+\alpha_l\varphi(v_l)\big)=0
			,
		\]
		and hence we find $\beta_1^\pm,\ldots,\beta_k^\pm\geq 0$ with 
		\[
			x-(\alpha_1v_1+\ldots+\alpha_lv_l)=(\beta_1^+u_1+\ldots+\beta_k^+u_k)+(\beta_1^-(-u_1)+\ldots+\beta_k^-(-u_k))
			.
		\]
	\end{Nili}
\end{Proof}

\subsection*{\ding{226}\ Reducing from \Vect{R}\ to \Pca}

The reduction lemma for passing from vector spaces over $\mathbb R$ to positively convex algebras is again a corollary of 
Theorem~\ref{PolyhedralR}. 
However, in a sense the situation is more complicated. 
One, the corresponding strong version fails; in fact, 
no (nonzero) $\mathbb R$-vector space is finitely generated as a \Pca. 
Two, unlike in categories of semimodules, the direct product $T_{[0,1]}B_1\times T_{[0,1]}B_2$ 
does not coincide with $T_{[0,1]}(B_1\dot\cup B_2)$. 

\begin{lemma}\label{RedLemPca}
	Let $n_1,n_2\in\mathbb N$, and let $Z$ be a linear subspace of $\mathbb R^{n_1}\times\mathbb R^{n_2}$. 
	Then $Z\cap(\Delta^{n_1}\times\Delta^{n_2})$ is finitely generated as a positively convex algebra.
\end{lemma}
\begin{Proof}
	Obviously, $Z$ and $\Delta^{n_1}\times\Delta^{n_2}$ are both polyhedral. 
	We conclude that $Z\cap(\Delta^{n_1}\times\Delta^{n_2})$ 
	is generated by a finite set of points and directions. Since it is bounded, no direction can occur, and 
	it is thus finitely generated as a \Pca.
\end{Proof}

\section{Self-contained proof of Lemma~\ref{RedLemN}}

We provide a short and self-contained proof of the named reduction lemma. 
It proceeds via an argument very specific for $\mathbb N$; the essential ingredient is that 
the order of $\mathbb N$ is total and satisfies the descending chain condition.
Note that the following argument also proves Hilbert's Theorem.

First, a common fact about the product order on $\mathbb N^m$ 
(we provide an explicit proof since we cannot appoint a reference).

\begin{lemma}\label{Incomparable}
	Let $m\in\mathbb N$, and let $M\subseteq\mathbb N^m$ be a set of pairwise incomparable elements. 
	Then $M$ is finite. 
\end{lemma}
\begin{Proof}
	Assume that $M$ is infinite, and choose a sequence $(a_n)_{n\in\mathbb N}$ of different elements of $M$. 
	Write $a_n=(\alpha_{n,1},\ldots,\alpha_{n,m})$. 
	We construct, in $m$ steps, a subsequence $(b_n)_{n\in\mathbb N}$ of $(a_n)_{n\in\mathbb N}$ with the property 
	that (we write $b_n=(\beta_{n,1},\ldots,\beta_{n,m})$)
	\begin{equation}\label{5}
		\forall k\in\{1,\ldots,m\}\Deli
		L_k=\sup_{n\in\mathbb N}\beta_{n,k}<\infty
		\ \vee\ 
		\beta_{0,k}<\beta_{1,k}<\beta_{2,k}<\cdots
	\end{equation}
	In the first step, extract a subsequence of $(a_n)_{n\in\mathbb N}$ according to the behaviour of the 
	sequence of first components $(\alpha_{n,1})_{n\in\mathbb N}$.
	If $\sup_{n\in\mathbb N}\alpha_{n,1}<\infty$, take the whole sequence $(a_n)_{n\in\mathbb N}$ as the subsequence. 
	If $\sup_{n\in\mathbb N}\alpha_{n,1}=\infty$, take a subsequence $(a_{n_j})_{j\in\mathbb N}$ with 
	\[
		\alpha_{n_0,1}<\alpha_{n_1,1}<\alpha_{n_2,1}<\cdots
		.
	\]
	Repeating this step, always starting from the currently chosen subsequence, we succesively extract subsequences which 
	after $l$ steps satisfy the property \eqref{5} for the components up to $l$. 

	Denote 
	\[
		I_1=\big\{k\in\{1,\ldots,m\}\mid \sup_{n\in\mathbb N}\beta_{n,k}<\infty\big\}
		,\quad 
		I_2=\big\{k\in\{1,\ldots,m\}\mid \sup_{n\in\mathbb N}\beta_{n,k}=\infty\big\}
	\]
	The map $n\mapsto(\beta_{n,k})_{k\in I_1}$ maps $\mathbb N$ into the finite set $\prod_{k\in I_1}\{0,\ldots,L_k\}$, 
	and hence is not injective. Choose $n_1<n_2$ with $\beta_{n_1,k}=\beta_{n_2,k}$, $k\in I_1$. 
	Since $\beta_{n_1,k}<\beta_{n_2,k}$, $k\in I_2$, we obtain $b_{n_1}\leq b_{n_2}$. 
	However, by our choice of the elements $a_n$, $b_{n_1}\neq b_{n_2}$. 
	Thus $M$ contains a pair of different but comparable elements. 
\end{Proof}

\begin{Proof}[of Lemma~\ref{RedLemN}]
	If $\varphi^{-1}(\mathbb N^m)=\{0\}$, there is nothing to prove. Hence, assume that $\varphi^{-1}(\mathbb N^m)\neq\{0\}$. 
	\begin{Nili}
	\item \emph{Step~1:}\ 
		We settle the case that $Z\subseteq\mathbb Z^m$ and $\varphi$ is the inclusion map. 
		Let $M$ be the set of minimal elements of $(Z\cap\mathbb N^m)\setminus\{0\}$. 
		From the descending chain condition we obtain
		\[
			\forall x\in(Z\cap\mathbb N^m)\setminus\{0\}\Deli \exists\,y\in M\Deli y\leq x
		\]
		By Lemma~\ref{Incomparable}, $M$ is finite, say $M=\{a_1,\ldots,a_l\}$. 
		Now we show that $M$ generates $Z$ as commutative monoid. Let $x\in Z$, and assume that 
		$x-\sum_{j=1}^l\alpha_ja_j\neq 0$ for all $\alpha_j\in\mathbb N$. By the descending chain condition, 
		the set of all elements of this form 
		contains a minimial element, say, $x-\sum_{j=1}^l\tilde\alpha_ja_j$. Choose $y\in M$ with 
		$y\leq x-\sum_{j=1}^l\tilde\alpha_ja_j$. Since $y\neq 0$, we have 
		$x-\sum_{j=1}^l\tilde\alpha_ja_j-y<x-\sum_{j=1}^l\tilde\alpha_ja_j$ and we reached a contradiction.
	\item \emph{Step~2:}\ 
		The kernel $\varphi^{-1}(\{0\})$ is, as a subgroup of the finitely generated abelian group $Z$, 
		itself finitely generated (remember here that $\mathbb Z$ is a Noetherian ring). 
		Let $\{u_1,\ldots,u_k\}$ be a set of generators of $\varphi^{-1}(\{0\})$ as abelian group. 
		Then $\{\pm u_1,\ldots,\pm u_k\}$ is a set of generators of $\varphi^{-1}(\{0\})$ as a 
		commutative monoid.

		By Step~1 we find $\{a_1,\ldots,a_l\}\subseteq\mathbb Z^m$ generating $\varphi(Z)\cap\mathbb N^m$ as a 
		commutative monoid. Choose $v_j\in Z$ with $\varphi(v_j)=a_j$, $j=1,\ldots,l$. 
		Then we find, for each $x\in Z$, a linear combination of the $v_j$'s with nonnegative integer coefficients
		such that 
		\[
			\varphi\Big(x-\sum_{j=1}^l\nu_jv_j\Big)=0
			.
		\]
		Hence, $\{\pm u_1,\ldots,\pm u_k\}\cup\{v_1,\ldots,v_l\}$ generates $\varphi^{-1}(Z)$ as commutative monoid. 
	\end{Nili}
\end{Proof}

\section{Properties of $\Gh$}
\label{app:D}

\begin{Proof}[of Lemma~\ref{Ghat-simple}]
	Here the inclusion ``$\supseteq$'' is obvious. For the reverse inclusion, let $(o,\phi)\in\Gh X$ and 
	choose $p_{a,j}$ and $x_{a,j}$ according to Definition~\ref{Ghat-def}. 
	Set $p_a=\sum_{j=1}^{n_a}p_{a,j}$. 
	If $p_a=0$, set $x_a=0$. If $p_a>0$, set $x_a=\sum_{j=1}^n\frac{p_{a,j}}{p_a}x_{a,j}$. 
	Then $x_a\in X$ and $f(a)=\sum_{j=1}^{n_a}p_{a,j}x_{a,j}=p_ax_a$.
\end{Proof}

\begin{lemma}\label{GhSurj}
	The functor $\Gh$ preserves surjective algebra homomorphisms.
\end{lemma}
\begin{proof}
	Let $X,Y$ be \Pca s, and $f:X\to Y$ a surjective algebra homomorphism. 
	Let $(o,g)\in\Gh Y$ be given. By Lemma~\ref{Ghat-simple} we can choose $p_a\in[0,1]$ and $y_a\in Y$ such that 
	$o+\sum_{a\in A}p_a\leq 1$ and $g(a)=p_ay_a$, $a\in A$. 
	Since $f$ is surjective, we find $x_a\in X$ with $f(x_a)=y_a$. 
	Let $h:A\to X$ be the function $h(a)=p_ax_a$. Again using Lemma~\ref{Ghat-simple}, we see 
	that $(o,h)\in\Gh X$. By our choice of $x_a$, it holds that $(\Gh f)(o,h)=(o,g)$. 
\end{proof}

\begin{Proof}[of Lemma~\ref{Ghat-Minko}]
	Let $(o,\phi)\in\Gh X$, and choose $p_a\in[0,1]$ and $x_a\in X$ as in Lemma~\ref{Ghat-simple}. 
	Then $\Min X(\phi(a))=p_a\Min X(x_a)\leq p_a$, and hence $o+\sum_{a\in A}\Min X(\phi(a))\leq 1$. 
	Further, $o\in[0,1]$, in particular $o\geq 0$. 

	Conversely, assume that $o\geq 0$ and $o+\sum_{a\in A}\Min X(\phi(a))\leq 1$. 
	Let $a\in A$. Set $p_a=\Min X(\phi(a))$, then $p_a\in[0,1]$ since $\sum_{a\in A}p_a\leq 1$. 
	To define $x_a$ consider first the case that $\Min X(\phi(a))=0$. 
	In this case $\phi(a)=0$ since $X$ is bounded, and we set $x_a=0$. 
	If $\Min X(\phi(a))>0$, set $x_a=\frac 1{\Min X(\phi(a))}\phi(a)$. 
	Since $X$ is closed, we have $x_a\in X$. In both cases, we obtained a representation $\phi(a)=p_ax_a$ with 
	$p_a\in[0,1]$ and $x_a\in X$. Clearly, $o+\sum_{a\in A}p_a\leq 1$, and we conclude that $(o,\phi)\in\Gh X$. 
\end{Proof}

\begin{Proof}[of Lemma~\ref{LinExt}]
	We build the extension in three stages.
	\begin{Nili}
	\item \ding{172}\ We extend $c$ to the cone generated by $X$:\ 
		Set $C=\bigcup_{t>0}tX$, and define $c_1\colon C\to V_2$ by the following procedure. 
		Given $x\in C$, choose $t>0$ with $x\in tX$, and set $c_1(x)=t\cdot c\big(\frac 1tx\big)$. 
		By this procedure the map $c_1$ is indeed well-defined.
		To see this, assume $x\in tX\cap sX$ where w.l.o.g.\ $s\leq t$. Then $\frac 1tx=\frac st\cdot\frac 1sx$. Since 
		$\frac st\leq 1$, it follows that $c(\frac 1tx)=\frac st c(\frac 1sx)$, and hence 
		$t\cdot c(\frac 1tx)=s\cdot c(\frac 1sx)$. 
		Let us check that $c_1$ is cone-morphism, i.e., that 
		\[
			c_1(x+y)=c_1(x)+c_2(y),\ x,y\in C,\qquad c_1(px)=pc_1(x),\ x\in C,p\geq 0
			.
		\]
		Given $x,y\in C$, choose $t>0$ such that $x,y,x+y\in tX$. Observe here that $C$ is a union of 
		an increasing family of sets. Then 
		\begin{align*}
			c_1(x+y)= &\, 2t\cdot c\Big(\frac 1{2t}(x+y)\Big)
			=2t\cdot c\Big(\frac 12\cdot\frac 1tx+\frac 12\cdot\frac 1ty\Big)
			\\[2mm]
			= &\, 2t\cdot\Big[\frac 12c\big(\frac 1tx\big)+\frac 12c\big(\frac 1ty\big)\Big]
			\\[2mm]
			= &\, t\cdot c\Big(\frac 1tx\Big)+t\cdot c\Big(\frac 1ty\Big)
			=c_1(x)+c_2(y)
			.
		\end{align*}
		Given $x\in C$ and $p>0$, choose $t>0$ with $x\in tX$. Then $px\in(pt)X$, and we obtain 
		\[
			c_1(px)=pt\cdot c\Big(\frac 1{pt}(px)\Big)=pt\cdot c\Big(\frac 1tx\Big)=pt\cdot\frac 1t c_1(x)
			=pc_1(x)
			.
		\]
		For $p=0$, the required equality is trivial. 
		Finally, observe that $c_1$ extends $c$, since for $x\in X$ we can choose $t=1$ in the definition of $c_1$. 
	\item \ding{173}\ We extend $c_1$ to the linear subspace generated by $C$:\ 
		Since $C$ is a cone, we have $\spn C=C-C$. We define $c_2\colon \spn C\to V_2$ by the following procedure. 
		Given $x\in\spn C$, choose $a_+,a_-\in C$ with $x=a_+-a_-$, and $c_2(x)=c_1(a_+)-c_2(a_-)$. 
		By this procedure the map $c_2$ is indeed well-defined.
		Assume $x=a_+-a_-=b_+-b_-$. Then $a_++b_-=b_++a_-$, and we obtain 
		\[
			c_1(a_+)+c_1(b_-)=c_1(a_++b_-)=c_1(b_++a_-)=c_1(b_+)+c_1(a_-)
			,
		\]
		which yields $c_1(a_+)-c_1(a_-)=c_1(b_+)-c_1(b_-)$. 
		Let us check that $c_2$ is linear.
		Given $x,y\in\spn C$, choose representations $x=a_+-a_-$, $y=b_+-b_-$. Then $x+y=(a_++b_+)-(a_-+b_-)$, 
		and we obtain 
		\begin{align*}
			c_2(x+y)= &\, c_1(a_++b_+)-c_1(a_-+b_-)
			=\big[c_1(a_+)+c_1(b_+)\big]-\big[c_1(a_-)+c_1(b_-)\big]
			\\
			= &\, \big[c_1(a_+)-c_1(a_-)\big]+\big[c_1(b_+)-c_1(b_-)\big]
			=c_2(x)+c_2(y)
			.
		\end{align*}
		Given $x\in\spn C$ and $p\in\mathbb R$, choose a representation $x=a_+-a_-$ and distinguish cases 
		according to the sign of $p$. 
		If $p>0$, we have the representation $px=pa_+-pa_-$ and hence 
		\begin{align*}
			c_2(px)= &\, c_1(pa_+)-c_1(pa_-)=pc_1(a_+)-pc_1(a_-)
			\\
			= &\, p\big[c_1(a_+)-c_1(a_-)\big]=pc_2(x)
			.
		\end{align*}
		If $p<0$, we have the representation $px=(-p)a_--(-p)a_+$ and hence 
		\begin{align*}
			c_2(px)= &\, c_1((-p)a_-)-c_1((-p)a_+)=(-p)c_1(a_-)-(-p)c_1(a_+)
			\\
			= &\, p\big[c_1(a_+)-c_1(a_-)\big]=pc_2(x)
			.
		\end{align*}
		For $p=0$, the required equality is trivial. 
		Finally, observe that $c_2$ extends $c_1$, since for $x\in C$ we can choose the representation 
		$x=x-0$ in the definition of $c_2$. 
	\item \ding{174}\ We extend $c_2$ to $V_1$:\ 
		By linear algebra a linear map given on a subspace can be extended to a linear map on the whole space. 
	\end{Nili}
	The uniqueness statement is clear. 
\end{Proof}

\begin{Proof}[of Corollary~\ref{Ghat-Minko-Coalg}]
	First assume that \eqref{2} holds. 
	Let $x\in X$. Then $\Min X(x)\leq 1$, and we obtain 
	\[
		\Out c(x)+\sum_{a\in A}\Min Y(\Comp ca(x))=
		\Out{\tilde c}(x)+\sum_{a\in A}\Min Y(\Comp{\tilde c}a(x))\leq\Min X(x)\leq 1
		.
	\]
	Further, $\Out c(x)\geq 0$ by assumption. Now Lemma~\ref{Ghat-Minko} gives $c(x)\in\Gh Y$. 

	Conversely, assume $c(X)\subseteq\Gh Y$, and let $x\in\mathbb R^n$ be given. 
	If $\Min X(x)=\infty$, the relation \eqref{2} trivially holds. 
	If $\Min X(x)=0$, then $x=0$ since $X$ is bounded. Hence, the left side of \eqref{2} equals $0$, and again \eqref{2} holds. 
	Assume that $\Min X(x)\in(0,\infty)$. Since $X$ is closed, we have $\Min X(x)^{-1}x\in X$, and hence 
	$c(\Min X(x)^{-1}x)\in\Gh Y$. From Lemma~\ref{Ghat-Minko}, we get the estimate 
	\begin{align*}
		\Out{\tilde c}(x)+\sum_{a\in A}\Min Y(\Comp{\tilde c}a(x))= &\, \Min X(x)
		\Big(
		\Out{\tilde c}\big(\frac 1{\Min X(x)}x\big)+\sum_{a\in A}\Min Y\big(\Comp{\tilde c}a(\frac 1{\Min X(x)}x)\big)
		\Big)
		\\
		= &\, \Min X(x)
		\Big(
		\Out c\big(\frac 1{\Min X(x)}x\big)+\sum_{a\in A}\Min Y\big(\Comp ca(\frac 1{\Min X(x)}x)\big)
		\Big)
		\leq\Min X(x)
		.
	\end{align*}
\end{Proof}

\section{Proof details of the Extension Theorem}

Recall Kakutani's theorem \cite[Corollary]{kakutani:1941}.

\begin{theorem}[Kakutani 1941]\label{Kak}
	Let $M\subseteq\mathbb R^n$ and $P\colon M\to\mathcal P(M)$. Assume 
	\begin{enumerate}
	\item $M$ is nonempty, compact, and convex,
	\item for each $x\in M$, the set $P(x)$ is nonempty, closed, and convex,
	\item the map $P$ has closed graph in the sense that, whenver $x_n\in M$, $x_n\to x$, and 
		$y_n\in P(x_n)$, $y_n\to y$, it follows that $y\in P(x)$.
	\end{enumerate}
	Then there exists $x\in M$ with $x\in P(x)$.
\end{theorem}

Note that $P$ having closed graph implies that $P(x)$ is closed for all $x$. To see this, let 
$y_n\in P(x)$, $y_n\to y$, and use the constant sequence $x_n=x$ in the closed graph property.

In the proof of Theorem~\ref{ExEx} we shall, as in Example~\ref{Pyramid}, identify a pyramid $Y$ with the appropriately scaled 
normal vector $u$ of its inclined side.
Then, for two pyramids $Y_1$ and $Y_2$ with corresponding normal vectors $u_1$ and $u_2$, 
the requirement that $X\subseteq Y_j$ becomes $(x,u_j)\leq\Min X(x)$, $x\geq 0$, and the requirement 
$\tilde c(Y_2)\subseteq\Gh Y_1$ becomes $\Out{\tilde c}(x)+\sum_{a\in A}(\Comp{\tilde c}a(x),u_1)\leq(x,u_2)$, $x\geq 0$, 
cf.\ Corollary~\ref{Ghat-Minko-Coalg}.

\begin{Proof}[of Theorem~\ref{ExEx}]
	Let $M$ be the set 
	\[
		M=\big\{u\in\mathbb R^n\mid u\geq 0\text{ and }(x,u)\leq\Min X(x),x\geq 0\big\}
		.
	\]
	We have to include vectors $u$ with possibly vanishing components into $M$ to ensure closedness. 
	It will be a step in the proof to show that a fixed point must be strictly positive.

	Let $P\colon M\to\mathcal P(M)$ be the map 
	\[
		P(u)=\big\{v\in M\mid \Out{\tilde c}(x)+\sum_{a\in A}(\Comp{\tilde c}a(x),u)\leq(x,v), x\geq 0\big\}
		.
	\]
	Here we again denote by $\tilde c\colon \mathbb R^n\to\mathbb R\times(\mathbb R^n)^A$ the linear extension of $c$.
	Observe that $\tilde c(x)\geq 0$ for all $x\geq 0$, since $\Delta^n\subseteq X$ and $c(x)\geq 0$ for $x\in X$. 

	It is easy to check that $M$ and $P$ satisfy the hypothesis of Kakutani's Theorem, 
	the crucial point being that $P(u)\neq\emptyset$.
	\begin{Nili}
	\item \ding{172}\ $M$ is nonempty:\ 
		We have $0\in M$.
	\item \ding{173}\ $M$ is compact:\ 
		To show that $M$ is closed let $u_n\in M$ with $u_n\to u$. 
		Since $u_n\geq 0$ also $u\geq 0$, and for each fixed $x\geq 0$ continuity of the scalar product yields 
		$(x,u)=\lim_{n\to\infty}(x,u_n)\leq\Min X(x)$.
		Further, $M$ is bounded since $(e_j,u)\leq\Min X(e_j)\leq 1$, $j=1,\ldots,n$, by our assumption that 
		$\Delta^n\subseteq X$, and hence $u\in[0,1]^n$.
	\item \ding{174}\ $M$ is convex:\ 
		Let $u_1,u_2\in M$ and $p\in[0,1]$. First, clearly, $pu_1+(1-p)u_2\geq 0$. Second, for each $x\geq 0$, 
		\begin{align*}
			(x,pu_1+(1-p)u_2)= &\, p(x,u_1)+(1-p)(x,u_2)
			\\
			\leq &\, p\Min X(x)+(1-p)\Min X(x)=\Min X(x)
			.
		\end{align*}
	\item \ding{175}\ $P(u)$ is nonempty:\ 
		Let $u\in M$ be given. The map $x\mapsto\Out{\tilde c}(x)+\sum_{a\in A}(\Comp{\tilde c}a(x),u)$ is a linear functional 
		on $\mathbb R^n$. Thus we find $v\in\mathbb R^n$ representing it as $x\mapsto(x,v)$. 
		Since $e_j\in X$, we have 
		\[
			(e_j,v)=\Out{\tilde c}(e_j)+\sum_{a\in A}(\Comp{\tilde c}a(e_j),u)\geq 0
			.
		\]
		Further, using that $u\in M$ and $\tilde c(X)\subseteq\Gh X$, we obtain that for each $x\geq 0$ 
		\[
			(x,v)=\Out{\tilde c}(x)+\sum_{a\in A}(\Comp{\tilde c}a(x),u)
			\leq\Out{\tilde c}(x)+\sum_{a\in A}\Min X(\Comp{\tilde c}a(x))
			\leq\Min X(x)
			.
		\]
		Together, we see that $v\in M$. By its definition, therefore, $v\in P(u)$. 
	\item \ding{176}\ $P(u)$ is convex:\ 
		Let $v_1,v_2\in P(u)$ and $p\in[0,1]$. First, since $M$ is convex, $pv_1+(1-p)v_2$ belongs to $M$. 
		Second, for each $x\geq 0$, 
		\begin{align*}
			(x,pv_1+ &\, (1-p)v_2)=p(x,v_1)+(1-p)(x,v_2)
			\\
			\geq &\, p\Big(\Out{\tilde c}(x)+\sum_{a\in A}(\Comp{\tilde c}a(x),u)\Big)+
			(1-p)\Big(\Out{\tilde c}(x)+\sum_{a\in A}(\Comp{\tilde c}a(x),u)\Big)
			\\
			= &\, \Out{\tilde c}(x)+\sum_{a\in A}(\Comp{\tilde c}a(x),u)
			.
		\end{align*}
	\item \ding{177}\ $P$ has closed graph:\ 
		Let $u_n\in M$, $u_n\to u$, and $v_n\in P(u_n)$, $v_n\to v$. Then $u,v\in M$ since $M$ is closed. Now fix $x\geq 0$. 
		Continuity of the scalar product allows to pass to the limit in the relation
		\[
			\Out{\tilde c}(x)+\sum_{a\in A}(\Comp{\tilde c}a(x),v_n)\leq(x,u_n)
			,
		\]
		which holds for all $n\in\mathbb N$. This yields $\Out{\tilde c}(x)+\sum_{a\in A}(\Comp{\tilde c}a(x),v)\leq(x,u)$.
	\end{Nili}
	Having verified all necessary hypothesis, 
	Theorem~\ref{Kak} can be applied and furnishes us with $u\in M$ satisfying $u\in P(u)$, explicitly, $u\in\mathbb R^n$ 
	with
	\begin{equation}\label{1}
		u\geq 0,\quad (x,u)\leq\Min X(x),x\geq 0,\quad \Out{\tilde c}(x)+\sum_{a\in A}(\Comp{\tilde c}a(x),u)\leq(x,u),x\geq 0
		.
	\end{equation}
	Set $Y=\{x\geq 0\mid (x,u)\leq 1\}$. Then $Y$ is a \Pca, and by definition contained in $\mathbb R_+^n$. 
	It contains $X$ since $u\in M$, and since $u\in P(u)$ we have $\tilde c(Y)\subseteq\Gh Y$. 
	Thus $d=\tilde c|_Y$ turns $Y$ into an $\Gh$-coalgebra, and since $c=\tilde c|_X=(\tilde c|_Y)|_X=d|_X$, 
	the inclusion map $\iota\colon X\to Y$ is an $\Gh$-coalgebra morphism. 

	It remains to show that $Y$ is generated by $n$ linearly independent vectors. Remembering again Example~\ref{Pyramid}, 
	this is equivalent to $u$ being strictly positiv. Let $I=\{j\in\{1,\ldots,n\}\mid (e_j,u)=0\}$. 
	For each $j\in I$ the last relation in \eqref{1} implies that 
	$\Out c(e_j)=0$ and $(\Comp ca(e_j),u)=0$, $a\in A$. Since $u\geq 0$ and $\Comp ca(e_j)\geq 0$, 
	we conclude that the vector $\Comp ca(e_j)$ can have nonzero components only in those coordinates 
	where $u$ has zero component. In other words, $\Comp ca(e_j)\in\spn\{e_i\mid i\in I\}$. 
	Now \eqref{nv} gives $I=\emptyset$. 
\end{Proof}

\section{Proof details for properness of $\Gh$}

\begin{Proof}[of Lemma~\ref{GhatRestr}]
	We denote $v_j=(1,\ldots,1)\in\mathbb R^{n_j}$. 
	By Example~\ref{Pyramid}
	\[
		\mu_{\Delta^{n_j}}(x_j)=(x_j,v_j),\ x_j\in\mathbb R_+^{n_j}
		.
	\]
	Since $(\Delta^{n_j},c_j)$ is an $\Gh$-coalgebra, Corollary~\ref{Ghat-Minko-Coalg} yields 
	\[
		\Out{\tilde c{_j}}(x_j)+\sum_{a\in A}(\Comp{\tilde c{_j}}a(x_j),v_j)\leq(x_j,v_j),\quad x_j\in\mathbb R_+^{n_j}
		,j=1,2
		.
	\]
	Summing up these two inequalities yields that for $x_1\in\mathbb R_+^{n_1}$ and $x_2\in\mathbb R_+^{n_2}$
	\begin{equation}\label{7}
		\Big[\Out{\tilde c{_1}}(x_1)+\Out{\tilde c{_2}}(x_2)\Big]+
		\sum_{a\in A}\Big[(\Comp{\tilde c{_1}}a(x_1),v_1)+(\Comp{\tilde c{_2}}a(x_2),v_2)\Big]
		\leq(x_1,v_1)+(x_2,v_2)
		.
	\end{equation}
	Recall that $Z$ denotes the linear subspace of $\mathbb R^{n_1}\times\mathbb R^{n_2}$ constructed 
	in the basic diagram (referring back to Lemma~\ref{Prod}). 
	The definition of the map $d$ in the basic diagram ensures that for $(x_1,x_2)\in Z$
	\[
		\Out d((x_1,x_2))=\Out{\tilde c{_1}}(x_1)=\Out{\tilde c{_2}}(x_2),\quad 
		\Comp da((x_1,x_2))=\big(\Comp{\tilde c{_1}}a(x_2),\Comp{\tilde c{_2}}a(x_2)\big)
		.
	\]
	Set $v=\frac 12(1,\ldots,1)\in\mathbb R^{n_1+n_2}$. 
	Plugging the above into \eqref{7} and dividing by $2$ yields 
	\[
		\Out d((x_1,x_2))+\sum_{a\in A}\big(\Comp da((x_1,x_2)),v)\leq \big((x_1,x_2),v\big)
		,\quad (x_1,x_2)\in Z\cap\mathbb R_+^{n_1+n_2}
		.
	\]
	We have 
	\[
		\mu_{Z\cap 2\Delta^{n_1+n_2}}(x)=\max\big\{\mu_Z(x),\mu_{2\Delta^{n_1+n_2}}(x)\big\}
		=
		\begin{cases}
			(x,v) &\hspace*{-3mm},\quad x\in Z\cap\mathbb R_+^{n_1+n_2},
			\\
			\infty &\hspace*{-3mm},\quad \text{otherwise}.
		\end{cases}
	\]
	Here the first equality holds by property 3.\ listed after Definition~\ref{Minko}.
	The second equality is based on Example~\ref{MinkoExa} and Example~\ref{Pyramid}:
	First, $2\Delta^{n_1+n_2}$ is the pyramid constructed with $v$, and hence 
	$\mu_{2\Delta^{n_1+n_2}}(x) = (x,v)$ if $x\geq 0$, and $\infty$ otherwise.
	Second, $Z$ is a linear subspace, hence in particular a convex cone, and thus
	$\mu_Z(x) = 0$ if $x\in Z$, and $\infty$ otherwise. 

	The inclusion $d(Z\cap 2\Delta^{n_1+n_2})\subseteq\Gh(Z\cap 2\Delta^{n_1+n_2})$ can now be deduced 
	with help of Lemma~\ref{Ghat-Minko}. 
	Start with $(x_1,x_2)\in Z\cap 2\Delta^{n_1+n_2}$. Then $((x_1,x_2),v)\leq 1$. Moreover, $(x_1,x_2)\geq 0$ and 
	$(x_1,x_2)\in Z$ which allows to apply $d$. We obtain $d_o(x_1,x_2) + \sum_a (d_a(x_1,x_2),v) \leq 1$. 
	Since $d(Z)\subseteq F_{\mathbb R}(Z)$, we have $d_a(x_1,x_2)\in Z$. Now remember the computation of $d_a(x_1,x_2)$. 
	The map $\tilde c_j$ is the linear extension of $c_j$, hence 
	\[
		\tilde c_j(\Delta^{n_j}) = c_j(\Delta^{n_j}) \subseteq \hat F(\Delta^{n_j}) \subseteq [0,1]\times [0,1]^A
		.
	\]
	In particular, $\tilde c_j$ takes nonnegative values on $\Delta^{n_j}$, and by linearity thus on all of $(\mathbb R_+)^{n_j}$.
	This shows $d_a(x_1,x_2)\geq 0$ and $d_o(x_1,x_2)\geq 0$.
	By the above computation of the Minkowski functional $\mu_{Z\cap 2\Delta^{n_1+n_2}}$, 
	by now we know that we are in the first case, $(d_a(x_1,x_2),v)=\mu_{Z\cap 2\Delta^{n_1+n_2}}(d_a(x_1,x_2))$, 
	and thus 
  	\[
		d_o(x_1,x_2) + \sum_a \mu_{Z\cap 2\Delta^{n_1+n_2}}(d_a(x_1,x_2)) \leq 1
		.
	\]
	Lemma 16 applies, and yields $d(x_1,x_2)\in \hat F(Z\cap 2\Delta^{n_1+n_2})$.
\end{Proof}

\begin{Proof}[of Lemma~\ref{GhatNodes}]
	Using the basic diagram, we obtain 
	\begin{align*}
		& \tilde c_j(\Delta^{n_j})\subseteq\Gh\Delta^{n_j}\subseteq\Gh Y_j,
		\\
		& \tilde c_j(\pi_j(Z\cap 2\Delta^{n_1+n_2}))\subseteq\Gh(\pi_j(Z\cap 2\Delta^{n_1+n_2}))\subseteq\Gh Y_j.
	\end{align*}
	Since $\tilde c_j$ is linear, in particular convex, and $\Gh Y_j$ is convex, it follows that 
	\[
		\tilde c_j\big(\co(\Delta^{n_j}\cup\pi_j(Z\cap 2\Delta^{n_1+n_2}))\big)\subseteq\Gh Y_j
		.
	\]
\end{Proof}

\begin{Proof}[Lemma~\ref{GhatFree}]
	We check that the \Pca\ $Y_j$ satisfies the hypothesis of Theorem~\ref{ExEx}. 
	By its definition $\Delta^{n_j}\subseteq Y_j\subseteq\mathbb R_+^{n_j}$. 
	Since $Y_j$ is finitely generated, recall that $Y_j$ is the convex hull of two finitely generated \Pca s, 
	it is a compact subset of $\mathbb R^{n_j}$. 
	Finally, since the coalgebra structure on $Y_j$ is an extension of the one on $\Delta^{n_j}$, the present assumption 
	\eqref{nv2} implies that the condition \eqref{nv} of Theorem~\ref{ExEx} is satisfied. 
	Note here that $\Delta^{n_j}\cap\spn\{e_i\mid i\in I\}=\co(\{e_i\mid i\in I\}\cup\{0\})$.

	Applying Theorem~\ref{ExEx} we obtain extensions $U_j$ as required.
\end{Proof}

\begin{Proof}[of Lemma~\ref{GhatRedLem}]
	We show that the diagram 
	\[
		\xymatrix@C=60pt{
			\Delta^n \ar[r]^f \ar[d]_c & X \ar[d]^g
			\\
			\Gh\Delta^n \ar[r]_{\id\times(f\circ -)} & [0,1]\times X^A
		}
	\]
	commutes. 
	First, for $k\not\in I$, we have $((\id\times(f\circ -))\circ c)(e_k)=(g\circ f)(e_k)$ 
	by the definition of $g$. Second, consider $k\in I$. Then $(g\circ f)(e_k)=0$ since $f(e_k)=0$. 
	By \eqref{6}, also $((\id\times(f\circ -))\circ c)(e_k)=0$. 
	
	Since $\Gh f$ maps $\Gh\Delta^n$ into $\Gh X$, we have $g(X)\subseteq\Gh X$. 
	This says that $X$ indeed becomes an $\Gh$-coalgebra with structure $g$. 
	Revisiting the above diagram shows that $f$ is an $\Gh$-coalgebra morphism. 
\end{Proof}

\begin{Proof}[of Corollary~\ref{GhatRedCor}]
	Applying Lemma~\ref{GhatRedLem} repeatedly, we obtain after finitely many steps an $\Gh$-coalgebra $(\Delta^k,g)$ 
	such that no nonempty subset $I\subseteq\{1,\ldots,k\}$ with \eqref{6} exists for $(\Delta^k,g)$, and that we have 
	an $\Gh$-coalgebra morphism $f\colon (\Delta^n,c)\to(\Delta^k,g)$.
	Note here that in each application of the lemma the number of generators decreases. 
\end{Proof}



\end{document}